\documentclass[a4paper,12pt]{amsart}

\usepackage[]{amsmath}
\usepackage{amsfonts}
\usepackage{amssymb}
\usepackage{xcolor}

\usepackage{tikz}
\usetikzlibrary{shapes.geometric,positioning}

\usepackage[margin=3cm]{geometry}

\DeclareMathOperator*{\argmin}{arg\,min}

\newtheorem{thm}{Theorem}[section]
\newtheorem{prop}[thm]{Proposition}
\newtheorem{lem}[thm]{Lemma}
\newtheorem{cor}[thm]{Corollary}
\newtheorem{conj}[thm]{Conjecture}

\theoremstyle{definition}
\newtheorem{defn}[thm]{Definition}

\newtheorem{rem}[thm]{Remark}
\newtheorem{observation}[thm]{Observation}
\newtheorem{example}[thm]{Example}

\newcommand{\lcut}{{\mathcal{LCUT}}}

\newcommand{\blue}[1]{\textcolor{black}{#1}}

\newcommand{\NNI}{{N\negthinspace N\negthinspace I}}
\newcommand{\TBN}{T\negthinspace B\negthinspace N}
\newcommand{\PN}{P\negthinspace N}
\newcommand{\TB}{T\negthinspace B}

\title[How tree-based is my network?]{How tree-based is my network? Proximity measures for unrooted phylogenetic networks}
\author{Mareike Fischer}
\author{Andrew Francis}
\address{Institute of Mathematics and Computer Science, University of Greifswald, Germany}
\email{email@mareikefischer.de}
\address{Centre for Research in Mathematics an Data Science, Western Sydney University, Australia}
\email{a.francis@westernsydney.edu.au}

\date{\today}
\begin{document}

\begin{abstract}
\blue{
Tree-based networks are a class of phylogenetic networks that attempt to formally capture what is meant by ``tree-like'' evolution.  A given non-tree-based phylogenetic network, however, might appear to be very close to being tree-based, or very far. In this paper, we formalise the notion of proximity to tree-based for unrooted phylogenetic networks, with a range of proximity measures.  These measures also provide characterisations of tree-based networks.  One measure in particular, related to the nearest neighbour interchange operation, allows us to define the notion of ``tree-based rank''.  This provides a subclassification within the tree-based networks themselves, identifying those networks that are ``very'' tree-based.  Finally, we prove results relating tree-based networks in the settings of rooted and unrooted phylogenetic networks, showing effectively that an unrooted network is tree-based if and only if it can be made a rooted tree-based network by rooting it and orienting the edges appropriately.  This leads to a clarification of the contrasting decision problems for tree-based networks, which are polynomial in the rooted case but NP complete in the unrooted. }

\end{abstract}
\maketitle

\section{Introduction} \label{s:intro}
Phylogenetic networks are a generalization of phylogenetic trees that allow for representation of reticulation events such as hybridization or horizontal gene transfer, and can also be used to represent uncertainty.  While \blue{horizontal gene transfer}, for instance, means that the evolutionary history of species cannot be adequately modeled by a tree, it has nevertheless been argued that evolution may be fundamentally ``tree-like'', in the sense that evolution basically follows a tree-like structure with occasional additional arcs here and there~\cite{martin2011early,francis2015phylogenetic}. 
This debate inspired the definition of a ``tree-based network'', which is a network that can be obtained from a tree by the addition of new arcs between arcs of the tree, initially defined in the rooted setting~\cite{francis2015phylogenetic}. This family of networks has received quite a bit of theoretical development over the last few years, including extensions to the non-binary case~\cite{jetten2016nonbinary,zhang2016tree}, and the unrooted setting~\cite{francis2018tree,hendriksen2018tree,fischer2018nonbinary}.  Numerous characterizations of rooted tree-based networks have been published, many using interesting combinatorial constructions such as graph matchings~\cite{francis2015phylogenetic,jetten2016nonbinary,zhang2016tree,francis2018new,fischer2018classes}.

Not all networks are tree-based of course, and it was observed~\cite[Corollary 1]{francis2015phylogenetic} that any network can be made tree-based by adding additional leaves, which means that a network connecting a set of taxa may fail to be tree-based simply because of incomplete sampling. Nevertheless, measuring the extent to which a network is tree-based is an important question, and brings all phylogenetic networks into the conversation about tree-basedness. In the rooted setting, i.e. for a rooted phylogenetic network $N^r$ with leaf set $X$, several definitions of distance from being tree-based were introduced in~\cite{francis2018new}, recently extended to the non-binary setting in~\cite{pons2018tree}:

\begin{enumerate}
	\item Any rooted spanning tree of $N^r$ has the leaves of $N^r$ 
	and possibly others: let $\ell(N^r)$ denote the minimal number of others. 
	\item If $N^r$ is not tree-based it can be made tree-based by adding extra leaves. Let $t(N^r)$ denote the smallest number of leaves needed. 
	\item {The vertex set } $V(N^r)$ of $N^r$ can be partitioned by a set of disjoint paths, of which there will be at least $|X|$ plus possibly some extras.  Let $p(N^r)$ denote the number of extras. 
\end{enumerate}
\blue{These rooted proximity measures are illustrated in Figure~\ref{f:1-leaf-rooted-eg}.}
All three are zero if and only if $N^r$ is tree-based, and have been shown to be generically equal to each other in the rooted setting. This provides several equivalent lenses through which to see non-tree-basedness in rooted networks~\cite{francis2018new,pons2018tree}.
\begin{figure}[ht]
	\includegraphics[]{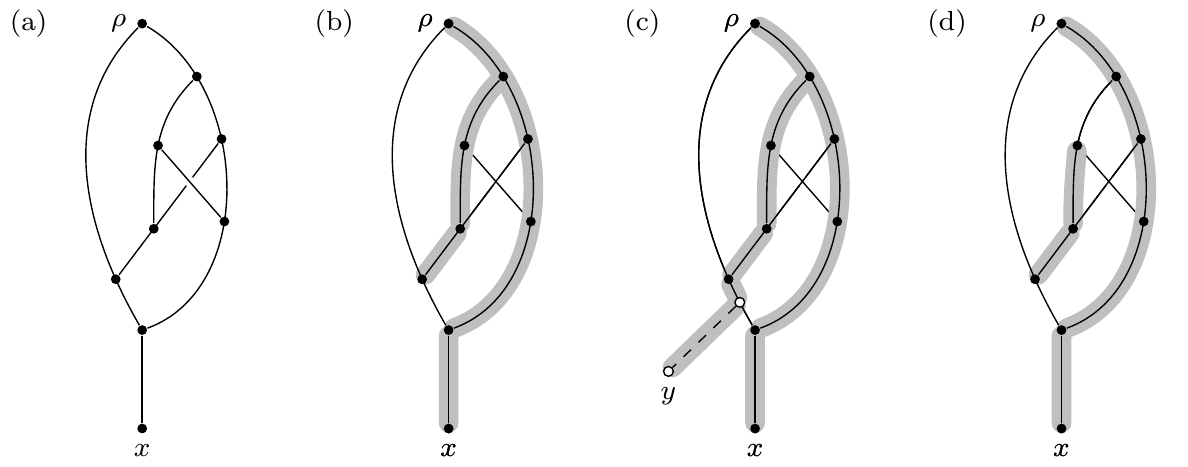}
	\caption{\blue{(a) A rooted phylogenetic network $N^r$ on one leaf $X=\{x\}$ that is not tree-based. (b) The same network with a spanning tree shown shaded, which has one additional leaf not from $X$, showing $\ell(N^r)=1$. (c) The network with an additional leaf $y$, making it tree-based, showing $t(N^r)=1$. (d) The network with vertices partitioned into 2 disjoint paths (one more than $|X|$), showing $p(N^r)=1$.}}\label{f:1-leaf-rooted-eg}
\end{figure}

In this paper, we define eight measures of distance from tree-basedness for \emph{unrooted} phylogenetic networks, both in the binary and non-binary settings. Some of these are straightforward adaptations from the rooted case, but others require subtle modifications.  Five of these are entirely new measures, in the sense that they are not unrooted versions of known rooted measures.  One of them can be shown to be equal to the generalizations, whereas the others can be shown to be distinct. 

Possibly the most interesting of the new measures \blue{introduced in Section~\ref{s:3.new.measures}}, namely $\delta_{\NNI}$, is based on so-called nearest neighbor interchange (NNI) moves. $\delta_{\NNI}$ measures the NNI distance of a non-tree-based network to the nearest tree-based one. 

As we explore in Section \ref{s:NNI.questions}, this new measure cannot only measure the distance from a non-tree-based network to a tree-based one (as the other measures also do), but additionally, it can measure the proximity of a tree-based network to a non-tree-based one. This leads to an interesting measure of robustness of tree-basedness for a given network \blue{$N$, that we call its \emph{tree-based rank}, $||N||_{\TB}$.  This rank takes non-negative values for tree-based networks, and negative values for non-tree-based networks.  Tree-based networks ``on the boundary'' of the space --- meaning that they are one NNI move from losing tree-based-ness --- have tree-based rank zero.} \blue{While it is now known that the space of tree-based networks is connected~\cite{fischer2019space}, this measure gives further structure to this space, and opens many further questions.}

\blue{The paper concludes with a focus on the connections between rooted and unrooted tree-based networks.  This exploration is along the lines of~\cite{huber2019rooting}, which seeks to understand what conditions on an unrooted phylogenetic network allow it to be made into a rooted phylogenetic network.  Here, we restrict to tree-basedness, and are able to obtain some characterizations that are clear, if not completely unsubtle.  For instance, Corollary~\ref{t:rooted-unrooted} provides a characterization for binary networks as follows:}  
\newtheorem*{root-unroot}{Corollary \ref{t:rooted-unrooted}}
\blue{
\begin{root-unroot} 
A {binary} unrooted phylogenetic network $N$ on $X$ is tree-based (in the unrooted sense) if and only if it can be rooted on the midpoint of an edge, and with orientations specified on the edges, to give a {binary} rooted phylogenetic network $N^r$ that is ``phylogenetically'' tree-based (i.e. it has a support tree whose root has out-degree 2). 
\end{root-unroot}
Theorems~\ref{t:degenerate.rooted-unrooted} and~\ref{t:strictly.rooted.unrooted} provide similar results for general unrooted tree-based networks, and for ``strictly'' tree-based networks, respectively. }

\blue{The paper is structured as follows.  We begin with background definitions for phylogenetic networks in Section~\ref{s:definitions}, and then introduce the first four proximity measures in Section~\ref{s:four.equal.measures}. These measures are all equal, and three of them are direct generalizations of the measures defined for rooted phylogenetic networks in~\cite{francis2018new}.  We then introduce an additional four measures in Section~\ref{s:3.new.measures}, one of which is the measure $\delta_\NNI$ related to the nearest neighbour interchange move, and mentioned already above.  The notion of tree-based rank is also introduced in this section, and further questions relating to it are discussed.  Section~\ref{s:summary} summarizes the eight proximity measures and what we know about their inter-relationships.  The final section of new content in the paper, Section~\ref{s:rooted}, looks back at the connection to rooted phylogenetic networks, and proves several results for different classes of networks along the lines of ``the unrooted phylogenetic network $N$ is tree-based if and only if there is a selection of place for a root, and an orientation on edges, that makes it a rooted tree-based phylogenetic network''.  }
\blue{We end in Section~\ref{s:discussion} with a summary of the main results and some further avenues for development.}

\section{Definitions and background}\label{s:definitions}

An \blue{\emph{unrooted phylogenetic network}} on a set $X$ (typically a set of species or taxa) is a connected graph without degree 2 vertices, whose degree 1 vertices (leaves) are {bijectively} labelled by the elements of $X$. In the following, whenever there is no ambiguity, we use the term \blue{\emph{network}} to refer to an unrooted phylogenetic network. 
\blue{Note that }we do not restrict the degree of the non-leaf vertices, so we deal with the general ``non-binary'' case.  \blue{That is, in this paper we allow the internal (non-leaf) vertices of the network to have any degree $>2$} (in contrast, the internal vertices of ``binary'' networks all have degree exactly 3)\footnote{{Note that the term ``non-binary'', which is generally used in the literature, is actually a bit misleading, as it summarizes phylogenetic networks whose internal vertices all have degree 3 (i.e. the binary ones) as well as those whose internal vertices have any degree $>2$. Thus, in particular, binary networks are ``non-binary'', too.}}.

The special case of an acyclic phylogenetic network is called a \blue{\emph{phylogenetic tree}}.

Throughout this manuscript, we will  assume that $|X|\ge 2$, and that $N$ is \emph{proper}. A proper network is one for which all components obtained by removing a cut edge or cut vertex contain at least one element of $X$ (following the definition of \cite{fischer2018nonbinary}, more general than the one given in~\cite{francis2018new}).  Note that a network that is not proper cannot be tree-based (cf. for instance ~\cite{fischer2018nonbinary}). {We will denote the set of proper unrooted phylogenetic networks on $X$ by $\PN(X)$.} 

\blue{We will need the concepts of ``deletion'' of an edge and ``suppression'' of a degree 2 vertex, defined as follows.  If $G$ is a graph with a vertex $v$ of degree 2 having neighbours $u,w$, so that $\{ u,v\}$ and $\{v,w\}$ are edges of $G$, then \emph{suppressing} $v$ means removing $v$ from the vertex set of $G$ and replacing the edges $\{u,v\},\{v,w\}$ by the edge $\{u,w\}$; that is, creating a new graph $G'$ whose vertex set is $V(G)\setminus\{v\}$, and whose edge set is $(E(G)\setminus\{\{u,v\},\{v,w\}\})\cup\{\{u,w\}\}$.  Note that $E(G)$ is a set, so in the event that $\{u,w\}$ was already an edge of $G$, this process results in a graph with edge set $E(G)\setminus\{\{u,v\},\{v,w\}\}$ (in particular, no parallel edges are created). If $G$ is a graph with edge set $E$, the graph $G'$ obtained from $G$ by \emph{deleting} the edge $e=\{u,v\}\in E$ has edge set $E'=E\setminus\{e\}$ and vertex set obtained from $V(G)$ by suppressing the vertices $u,v$ if they have become degree 2.  }

If $k$ is minimal such that the deletion of ${k}$ edges of ${N}$ would turn ${N}$ into a tree (i.e. a connected acyclic graph), we say that ${N}$ has \emph{tier} ${k}$. Note that the tier does not depend on $N$ being a phylogenetic network -- in fact, the tier of a connected graph can be defined analogously, and for technical reasons, we need this later on in this manuscript.

A related concept that we need to introduce is the {\it level} of a phylogenetic network. In this regard, recall that a {\it blob} of a network (or, more generally, of a graph) is a maximal connected subgraph that has no cut edge (if such a blob consists of only one vertex, it is called {\it trivial}). {Note that while in a binary phylogenetic network, i.e. a network in which all internal vertices have degree 3, blobs cannot contain any cut vertices (as all cut vertices in a binary network are incident to a cut edge), a blob in a non-binary phylogenetic network is \blue{explicitly} allowed to contain cut vertices (see \cite{fischer2018nonbinary} for more details).} A phylogenetic network (or a graph) is called {\it simple} if it contains at most one non-trivial blob. Now, a phylogenetic network (or graph) $N$ is said to have {\it level} $k$, 
if the maximal tier of the blobs of $N$ is $k$ (consequently, for any network $N$, $level(N)\le tier(N)$).

While the tier and the level of $N$ are related concepts, they may be arbitrarily different. 
For instance, the network in \blue{Figure~\ref{f:level-5-example} has two blobs, and is level 5 and tier $10$. The tier can be raised arbitrarily by additions of further blobs.}

\begin{figure}[ht]
\includegraphics[]{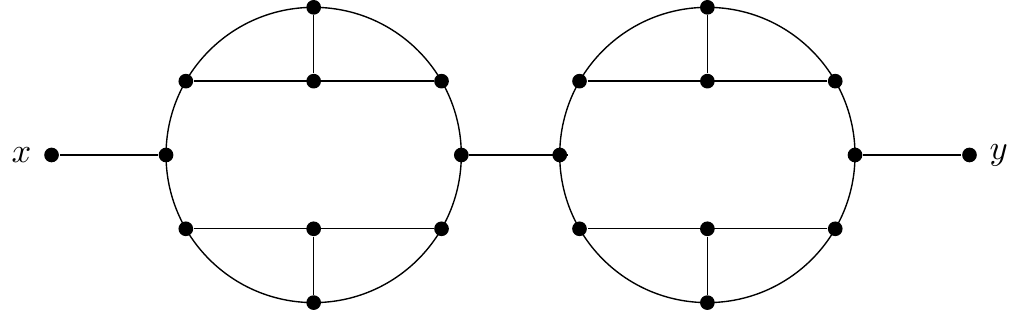}
\caption{\blue{A phylogenetic network that is level 5 and tier 10.}}\label{f:level-5-example}
\end{figure}

\blue{In the following, we denote by $V^d(N)$ the set of degree $d$ vertices in $N$, so that $V^1(N)=X$ is the set of leaves.  A \emph{spanning tree} of a graph is an acyclic subgraph that contains all the vertices of the graph.}
A \emph{support tree} $T$ of a network $N$ is a spanning tree of $N$ satisfying $V^1(T)=V^1(N)=X$, that is, whose leaf set coincides with the leaf set $X$ of $N$ \blue{(see Figure~\ref{f:level5}(ii) for an example of a support tree)}.  If $N$ contains such a support tree $T$, it is called \emph{tree-based}. This is a direct generalization of the definition for binary phylogenetic networks in~\cite{francis2018tree}, and coincides with the notion of ``loosely'' tree-based introduced in~\cite{hendriksen2018tree}. Note that $T$ is not necessarily a phylogenetic tree as it may contain degree-2 vertices. We denote the set of tree-based networks on $X$ by $\TBN(X)$.  

As a side note, observe that if $N$ is a (proper) network with a spanning tree $T$ with exactly two leaves, then it is tree-based. This is because if one leaf of $T$ is not from $X$, then $N$ has only one leaf, in which case it consists of only one vertex (proper networks with one leaf are trivial~\cite[Remark 1]{fischer2018nonbinary}). Therefore both leaves are from $X$, and so $N$ has a support tree and is tree-based.  Furthermore, all proper {binary} networks of level less than or equal to 4 are also tree-based~\cite{francis2018tree}, {and the same is true for all proper non-binary networks of level less than or equal to 3~\cite{fischer2018nonbinary}}.

Generalizing support trees, we introduce the notion of a \emph{support network} $\widehat{N}$ of $N$ to be a connected subgraph of $N$ containing all of the vertices of $N$ and a subset $\widehat{E}\subseteq{E(N)}$ such that the leaf set of $\widehat{N}$ coincides with the leaf set $X$ of $N$. Note that a support network $\widehat{N}$ need not necessarily be a phylogenetic network, as it may contain vertices of degree 2.  \blue{An example of a support network is given in Figure~\ref{f:support.network}.} 

\begin{figure}[ht]
\includegraphics[]{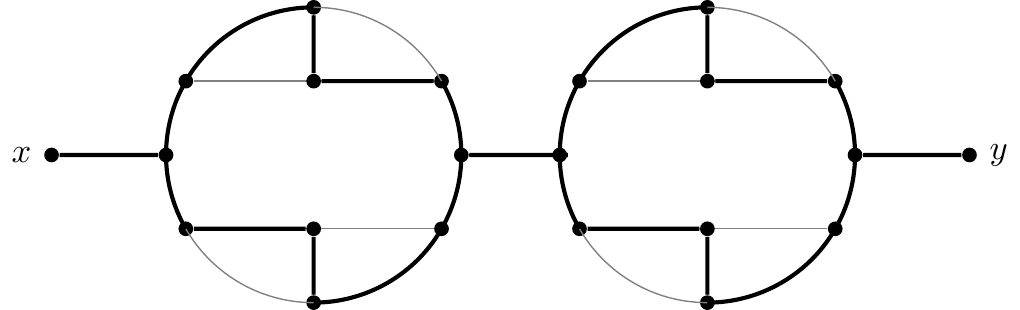}
\caption{\blue{A support network for the phylogenetic network shown in Figure~\ref{f:level-5-example}, shown in bold.} }\label{f:support.network}
\end{figure}

We call a phylogenetic network $N$ {\it tier-${k}$-based}, if it contains a support network $\widehat{N}$ of tier ${k}$. Note that if $N$ is tier-$0$-based, then $\widehat{N}$ is a support tree of $N$, and $N$ is tree-based. 
Moreover, note that if $N$ is a tier-$k'$ network which is tier-$k$-based, then $k\leq k'$, as $N$ can act as its own support network.

Similarly, we call a phylogenetic network $N$ {\it level-${k}$-based}, if it contains a support network $\widehat{N}$ of level ${k}$. Note that if $N$ is level-$0$-based, then $\widehat{N}$ is a support tree of $N$, and $N$ is tree-based. 
As for the tier, note that if $N$ is a level-$k'$ network which is level-$k$-based, then $ k\leq k'$, as $N$ can act as its own support network.
\blue{For example, the support network shown in Figure~\ref{f:support.network} is level 1 and tier 2, so this network is level-1-based and tier-2-based.}

In the course of this paper we will need the idea of adding edges and leaves to a network. This requires ``subdivision'' of edges.  

{An edge $\{u,v\}$ of a phylogenetic network may be \emph{subdivided} by first creating a new vertex $w$, then deleting $\{u,v\}$ and adding in $\{u,w\}$ and $\{w,v\}$.  That is, subdividing $\{u,v\}$ amounts to creating a new network $N'$ from $N$ whose vertex set is $V(N')=V(N)\cup\{w\}$ and whose edge set is $E(N')=(E(N)\setminus\{\{u,v\}\})\cup\{\{u,w\},\{v,w\}\}$.}

To add an internal edge to a phylogenetic network, first subdivide two edges $\{u_1,v_1\}$ and $\{u_2,v_2\}$, creating new vertices $w_1$ and $w_2$ respectively, and then add a new edge $\{w_1,w_2\}$. To add a leaf to an edge $\{u,v\}$, we subdivide it creating a new vertex $w$, then add an additional new vertex $x$ (the leaf) and the edge $\{w,x\}$.

Moreover, we will refer to the \emph{leaf cut graph} $\lcut(N)$ of a proper network $N\in\PN(X)$, with $|V(N)|\ge 3$, 
which is the graph $G$ obtained from $N$ by deleting all leaves and their incident edges~\cite{fischer2018classes}. Note that this may result in some vertices of degree 2 and -- e.g. if $N$ is a tree -- even new leaves not labelled by $X$, which we do \emph{not} remove.  We will also consider the \emph{simplified} $\mathcal{LCUT}$ graph $\mathcal{LCUT^{\text{simp}}}(N)$, which results from repeating the leaf deletion (even of leaves not from $X$) and -- if applicable -- suppressing degree-2 vertices until no such operation is possible anymore.

\blue{A \emph{path} in a network $N$ is a sequence of distinct vertices $v_1,v_2,\dots,v_k\in V(N)$ such that $\{v_i,v_{i+1}\}\in E(N)$ for all $1\le i<k$.}

The final concept we need to introduce for this manuscript is the \emph{nearest neighbor interchange (NNI)}, which is a replacement of a path in the network with an alternative path, defined as follows for unrooted phylogenetic networks in~\cite{huber2016transforming}.  

\begin{defn}[Nearest Neighbor Interchange (NNI)]
Let $N$ be a phylogenetic network in which $a,b,c,d$ is a path for which neither $\{a,c\}$ nor $\{b,d\}$ is an edge.  The NNI operation on this path replaces it with the path $a,c,b,d$: the edges $\{a,b\}$ and $\{c,d\}$ are deleted, and edges $\{a,c\}$ and $\{b,d\}$ are added. 
The NNI distance $d_\NNI(N_1,N_2)$ between two unrooted phylogenetic networks $N_1$ and $N_2$ is defined to be the minimum number of NNI moves \blue{necessary to obtain one from the other}.
\end{defn}

We now turn to proximity measures.

\section{Unrooted proximity measures}\label{s:four.equal.measures}

We wish to analyze the proximity of phylogenetic networks to being tree-based, and in this regard, we define the following measures. Measures (1)--(3) are analogous to the ones presented in~\cite{francis2018new} for rooted networks; we will discuss this relationship more in Section~\ref{s:rooted}. However, (4) gives rise to a new measure based on the tier of a network, which we later show to be identical to the ones given by (1)--(3) (cf. Theorem \ref{t:equalities}), so the tier provides a new perspective on the other measures. In Section \ref{s:new.measure.e} we will consider four additional new measures which can be shown to be different from the ones presented in this section.

\begin{defn}\label{d:four.measures}
Let $N$ be an unrooted phylogenetic network of level $k$.  
\begin{enumerate}
	\item Let $\ell(N):=\min\{|V^1(T)\setminus X|\mid T\text{ a spanning tree for }N\}$.
	\item Let $t(N):=\min\{|\widehat{X}| \mid \widehat{X}$ is a set of leaves that can be added to \blue{edges in} $N$ in order to $\left. \text{make $N$ tree-based}\right\}$; 
	\item Let $p(N):=k-|X|+1$, where $k$ is minimal such that there exists a sequence of paths $(\pi_1,\dots,\pi_k)$ in $N$ partitioning $V(N)$ and satisfying the property that for all $i=2,\ldots,k$, an endpoint of $\pi_i$ is adjacent to one or more of the paths $\pi_1,\dots,\pi_{i-1}$.
	\item Let $\tau(N):=\min\limits \{{k}\mid \mbox{$N$ is tier-${k}$-based}\}$. 
\end{enumerate}
\end{defn}

We are now in the position to prove the main theorem of this section.

\begin{thm}\label{t:equalities}
If $N\in \PN(X)$ 
and $|X|\ge 2$, then $\ell(N)=p(N)=t(N)=\tau(N)$.  

In particular, $N$ is tree-based if and only if $\ell(N)=t(N)=p(N)=\tau(N)=0$.
\end{thm}

\begin{proof}
We will show that $\ell(N)\leq p(N) \leq t(N) \leq \tau (N) \leq \ell(N)$. 

We begin with $\ell(N)\le p(N)$.  Let $(\pi_1,\dots,\pi_k)$ be a sequence inducing $p(N)$, i.e.  partitioning $V(N)$ such that an endpoint of $\pi_i$ is adjacent to one or more of the paths $\pi_1,\dots,\pi_{i-1}$ for all $i=2,\ldots,k$. We construct a spanning tree $T$ of $N$ as follows: start with $\pi_1$, and attach the path $\pi_i$ for all $i=2,\ldots,k$ according to the order induced by the sequence. To be precise, given $\pi_1$, we know that one endpoint of $\pi_2$ is adjacent to $\pi_1$, so we only need to add one edge from $N$ in order to connect both paths, and this connection will result in a tree with at most 3 leaves (as one endpoint of $\pi_2$ is not a leaf of the resulting tree, but both leaves of $\pi_1$ might be). We subsequently repeat this for following $\pi_i$, in each step adding at most one leaf to the spanning tree. This results in a spanning tree $T$ of $N$ with at most $k+1=p(N)+|X|$ leaves (at most 2 from $\pi_1$ and at most 1 for all other $i\in \{2,\ldots,k\}$). 
So clearly, a spanning tree of $N$ with the minimum number of leaves will have at most $p(N)+|X|$ leaves. By definition of $\ell(N)$, this implies $\ell(N)+|X|\leq p(N)+|X|$ and thus $\ell(N)\leq p(N)$. 

Next we show that $p(N)\le t(N).$

Suppose that $t(N)$ leaves are added to $N$ to make a tree-based network $N'$, and that this is minimal.  Then there is a spanning tree $T'$ of $N'$ whose leaves are the $|X|+t(N)$ leaves of $N'$.  
We will use $T'$ to construct a set of paths that partition the vertices of $N'$, as follows.  

Choose a path $\pi_1$ between two leaves in $T'$ that are also leaves of $N$ (recall that by assumption, $|X|\geq 2$).  Delete all edges $\{u_i,v_i\}$ of $T'$ that are not in the path but for which one endpoint (say $u_i$) is on the path.  What remains apart from the path itself are a set of trees $T_i$ with one vertex $v_i$ that is one edge in $N'$ distant from the path, and whose leaves are all leaves of $N'$.  For each such tree $T_i$, choose a path in $T_i$ from $v_i$ to a leaf, number the paths arbitrarily from $\pi_2$ onwards, and as before, delete all edges in $T_i$ that have one vertex in the path and the other not.  This again creates a set of sub-trees of $T_i$ for each $i$.  This process may be repeated until what remains is an ordered set of paths partitioning $V(N')$, each of which contains exactly one leaf of $N'$ except for the first path $\pi_1$, which contains two leaves.  That is, the number of such paths is one less than the number of leaves in $N'$, namely $|X|+t(N)-1$. These paths also satisfy the criterion of the definition of $p(N)$, namely that aside from the first, they have one endpoint a leaf on $N'$, and the other is adjacent to one of the earlier paths in the sequence.

From each of these paths whose corresponding leaf endpoint is not from $N$, delete the edge ending in the leaf.  Now we have $|X|+t(N)-1$ paths that partition $V(N)$, each of which (apart from the first) has one endpoint adjacent to a preceding path, as required by the definition of $p(N)$.  Some of these paths may of course be empty, and it may be that there is a more optimal choice of paths that does the job.  But in any case the minimal number of paths $k$ satisfies $k\le |X|+t(N)-1$, and so $p(N)=k-|X|+1\le t(N)$ as required.

Now we claim $t(N)\le\tau(N)$.  First, observe that if $\tau(N)=0$ then $N$ is tree-based and so $t(N)=0$, and the inequality holds.
So suppose now that we have a support network $\widehat N$ for $N$ of minimal tier $\tau(N)>0$.  We are going to add leaves to $\widehat N$ that in each case reduce the tier by at least 1.  Consider a cycle in $\widehat N$.  All cycles in $\widehat N$ must have vertices of degree 2, since otherwise the cycle would not be needed to cover all vertices, violating minimality of the tier of $\widehat N$.  Choose a degree 2 vertex $v$ in the cycle adjacent to a vertex $w$ of higher degree in $\widehat N$ (which must exist as otherwise the cycle would not be connected to the rest of $\widehat N$, contradicting the connectedness of support networks), add a new vertex $v'$ to the edge $\{v,w\}$, add a leaf to $v'$, and delete the edge $\{v',w\}$, from both $N$ and $\widehat N$. 

This creates a new network $N'$ and support network $\widehat N'$, but now $tier(\widehat N')<tier(\widehat N)$ because one cycle has been eliminated.  Thus at most $tier(\widehat N)$ additional leaves need to be added to $N$ to make a network that has a support network of tier 0, and is therefore tree-based.  It follows that $t(N)\le\tau(N)$.

Finally, we show that $\tau(N)\le\ell(N)$.
 
Let $T$ be a spanning tree of $N$ that has $\ell(N)$ extra leaves, i.e. leaves which are internal vertices of $N$ (additional to leaf set $X$). Each such extra leaf of $T$ has degree 1 in $T$ but degree at least 3 in $N$. In particular, each such vertex has an incident edge which is not contained in $T$. For each extra leaf, we add one such edge of $N$ to $T$, i.e. for each leaf of $T$,
we increase the tier by 1. The result is a support network $\widehat{N}$ of $N$ of tier $\ell(N)$. Therefore, by the minimality of $\tau(N)$, $\tau(N)\leq \ell(N)$. 

This completes the proof. 
\end{proof}

\begin{rem}
There is a key difference between the rooted and unrooted measures in the path-based measure $p(N)$, which does not directly generalize from rooted to unrooted.  A direct generalization might be a measure such as ``the number of edge disjoint paths that partition the network, minus $|X|$'', but this does not hold for unrooted networks, as can be seen in the level 5 example in Figure~\ref{f:level5}(i), which is not tree-based but whose vertex set can be partitioned by $|X|=2$ edge-disjoint paths.  The point is that paths in rooted phylogenetic networks are directed, and the information carried in that directedness needs to be captured in the unrooted characterization by the implied ordering on the paths in Definition~\ref{d:four.measures}(3).
\begin{figure}[ht] 
\includegraphics[width=\textwidth]{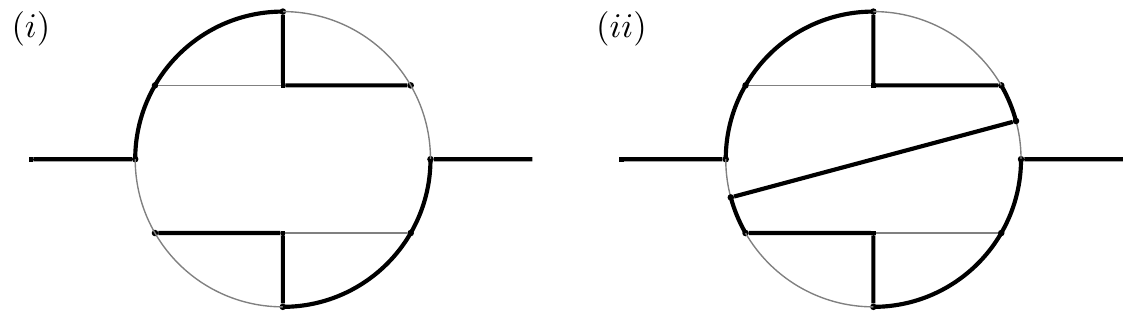}
\caption{$(i)$ A level 5 non-tree-based network whose vertices are covered by two edge-disjoint paths, shown in bold (this example of a non-tree-based network was independently presented in~\cite{fischer2018nonbinary} and the erratum \cite{francis2019correction} to~\cite{francis2018tree}).
$(ii)$ The same network with an additional diagonal edge  making it tree-based (support tree highlighted in bold). This example is referred to in Theorem~\ref{t:e(N)}.
}\label{f:level5}
\end{figure}
\end{rem}

It can be easily shown that, while the rooted counterparts of $\ell$, $t$, and $p$ introduced in \cite{francis2018new} can all be calculated in polynomial time, this is not possible for the above unrooted proximity measures, because otherwise, it could easily be decided if they equal 0 (which would contradict the known NP-completeness of the unrooted tree-basedness decision problem~\cite{francis2018tree}, cf. Section \ref{s:discussion}). However, if it is already known that the $\lcut$ graph or the $\lcut^{\text{simp}}$ graph of an unrooted phylogenetic network $N$ have a Hamiltonian path\footnote{\blue{A Hamiltonian path in a graph is a path that visits every vertex in the graph exactly once.}}, all four measures in Definition~\ref{d:four.measures} are bounded, which we prove in the following proposition.  

\begin{prop}\label{p:lcut.Hamiltonian}
Let $N \in PN(X)$. Then, if $\lcut(N)$ or $\lcut^{\text{simp}}(N)$ has a Hamiltonian path then $\ell(N)\le 2$.
\end{prop}

\begin{proof}
There is a spanning tree for the network that consists of the Hamiltonian path plus the additional edges to the leaves from $N\setminus\lcut(N)$ or plus the additional pending subtrees induced by $N\setminus\lcut^{\text{simp}}(N)$, respectively, and this spanning tree has leaves that are all elements of $X$ but possibly in addition one or both of the endpoints of the Hamiltonian path. This completes the proof.
\end{proof}

Note that Proposition \ref{p:lcut.Hamiltonian} nicely extends a result from \cite{fischer2018classes}, which states that if $\lcut(N)$ is Hamilton-connected (i.e. if for all pairs $\{a,b\}$ from the vertex set of $\lcut(N)$ there is a Hamiltonian path from $a$ to $b$), then $N$ is tree-based, i.e. $\ell(N)=0$. 

We next turn our attention to four new proximity measures.

\section{Four new proximity measures}\label{s:new.measure.e}\label{s:3.new.measures}

{In this section we introduce four new proximity measures to being tree-based. The first is similar to the tier-based concept introduced in the previous section but uses the level of the network instead of the tier, the second uses the number of edges one needs to add to a network to make it tree-based, the third uses the number of nearest neighbor interchange (NNI) moves, 
and the fourth counts the excess edges present in a spanning tree. For all four of these measures, we will show that they are really new in the sense that they are not identical to the ones presented in the previous section.}

\subsection{A proximity measure based on the level of the network}

Let $N\in \PN(X)$, and define 
\[\lambda(N):=\min\limits \{{k}\mid \mbox{$N$ is level-${k}$-based}\}.\]
Note that $\lambda(N)$ is 0 precisely if $N$ is tree-based, because trees are the only level-0 networks, so if $\lambda(N)=0$, we know that $N$ has a support tree. 

The following main theorem of this subsection shows that the four measures described in Section~\ref{s:four.equal.measures} provide an upper bound for $\lambda(N)$. 

\begin{thm}\label{t:lambda.le.tau}
If $N$ is a phylogenetic network, then $\lambda(N) \leq \tau(N)$. 
\end{thm}

\begin{proof} 
As noted already in Section~\ref{s:definitions}, the level of a network is bounded above by its tier.

So now let $N$ be a phylogenetic network with support networks $N'$ and $N''$, where $N'= \argmin\limits \{{k}\mid \mbox{$N$ is level-${k}$-based}\}$ and $N''= \argmin\limits \{{k}\mid \mbox{$N$ is tier-${k}$-based}\}$. This implies that $N'$ gives us $\lambda(N)$ and $N''$ gives us $\tau(N)$. Let $level(N')$ and $level(N'')$ denote the levels of $N'$ and $N''$ and $tier(N')$ and $tier(N'')$ denote the tiers of $N'$ and $N''$, respectively.

Then, by the minimality of $N'$, we know that $\lambda(N)=level(N')Ê\leq level (N'')$, because both networks are support networks of $N$. On the other hand, as stated above, we know that $level(N'') \leq tier(N'') = \tau(N)$. So in total, this shows that $\lambda(N) \leq \tau(N)$ and thus completes the proof.
\end{proof}

Observe that the inequality stated in Theorem \ref{t:lambda.le.tau} can indeed be strict for some networks. In fact, $\tau$ and $\lambda$ can be made arbitrarily different, as in Observation~\ref{obs:tau.lambda}.

\begin{observation}\label{obs:tau.lambda}
Consider the network $N$ depicted in Figure \ref{taulambda_fig}. This network basically consists of $m$ non-trivial blobs, each of which is a copy of the simple network presented in Figure \ref{f:level5}(i). As the network from Figure \ref{f:level5}(i) is not tree-based, both the level and the tier of all support networks of each blob in $N$ must be at least 1. In fact, it can easily be seen that both the level and the tier of the support network of each individual blob in $N$ is precisely 1; cf. the highlighted support network in Figure \ref{taulambda_fig} (note that the equality of the level and the tier in this case are due to Proposition \ref{taulambda_simple} below, as the network in Figure \ref{f:level5}(i) is simple). 
However, this implies that $N$, the network from Figure \ref{taulambda_fig}, is level-1-based, while it is also clear that the tier of the support network of $N$ equals the number of such blobs of $N$, which is $m$. So we have $\lambda(N)=1$ and $\tau(N)=m$, which means that the difference $\tau(N)-\lambda(N)=m-1$ can be made arbitrarily large by adding more and more identical blobs. So in fact, for general networks $N$, $\lambda(N)$ and $\tau(N)$ need not be identical, which shows that $\lambda(N)$ is indeed a proximity measure different from the ones introduced in the previous section.

\begin{figure}[ht]
\includegraphics[width=\textwidth]{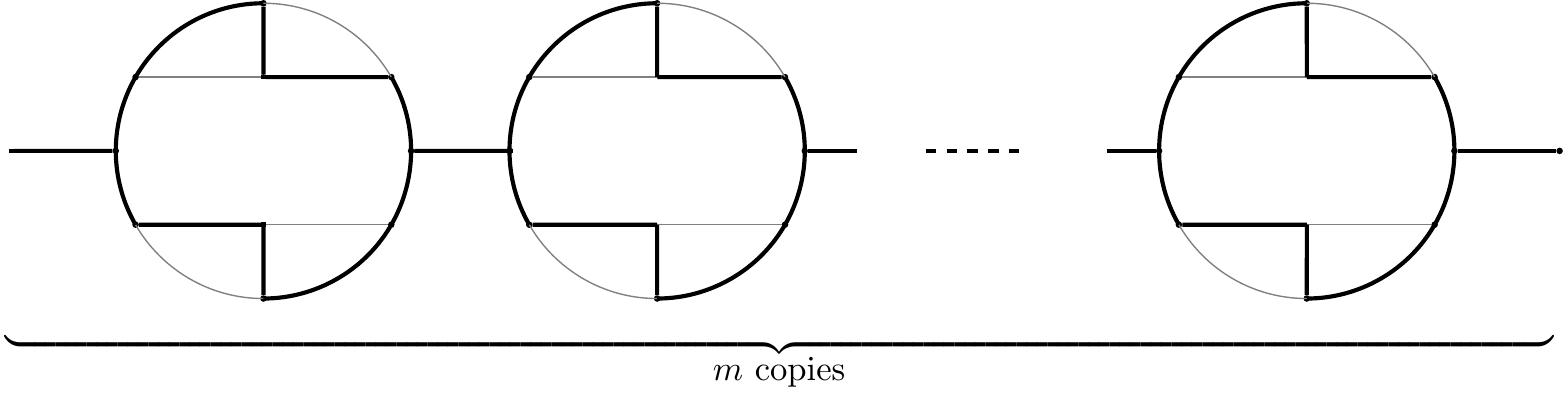}
\caption{A network $N$ consisting of $m$ blobs that are based on the simple non-tree-based network presented in Figure \ref{f:level5}(i). Here, we have $\lambda(N)=1$ and $\tau(N)=m$, which can be seen considering the highlighted support network. In particular, this shows that $\lambda(N)\neq \tau(N)$. }\label{taulambda_fig}
\end{figure}

\end{observation}

The fact that $\lambda$ and $\tau$ are closely related but generally not equal immediately gives rise to the question of whether there are types of networks for which we can guarantee equality. We end this section with the following proposition, which shows that indeed equality holds for simple networks \blue{(those with at most one non-trivial blob)}.

\begin{prop}\label{taulambda_simple}
If $N$ is a simple phylogenetic network, then $\lambda(N)=\tau(N)$.
\end{prop}

\begin{proof} Note that in simple networks $N$, by definition the level of $N$ equals the tier of $N$.

So now let $N$ be a simple network with support networks $N'$ and $N''$, where $N'= \argmin\limits \{{k}\mid \mbox{$N$ is level-${k}$-based}\}$ and $N''= \argmin\limits \{{k}\mid \mbox{$N$ is tier-${k}$-based}\}$. This implies that $N'$ gives us $\lambda(N)$ and $N''$ gives us $\tau(N)$. Let $level(N')$ and $level(N'')$ denote the levels of $N'$ and $N''$ and $tier(N')$ and $tier(N'')$ denote the tiers of $N'$ and $N''$, respectively. Note that as $N$ was simple, so are $N'$ and $N''$, which -- by the above observation -- implies $level(N')=tier(N')$ and $level(N'')=tier(N'')$. Moreover, by the minimality of $N''$, we can conclude $\tau(N)=tier(N'') \leq tier(N')=level(N')=\lambda(N)$. So $\tau(N) \leq \lambda(N)$. Moreover, by Theorem \ref{t:lambda.le.tau} we already know that $\lambda(N) \leq \tau(N)$. This completes the proof.

\end{proof}

\blue{Note that the reverse implication does not hold, by the following argument. 
If $N$ satisfies $\lambda(N)=\tau(N)=k$, then there is a tier $k$ support network $N'$ of $N$.  Thus, the levels of the blobs in $N'$ add up to $k$, and so if $N'$ has more than one blob, the maximum level of any blob in $N'$ must be strictly less than $k$.  But then the level of $N'$ would also be $<k$, meaning that $N'$ is a support network of $N$ whose level is $<k$: a contradiction to $\lambda(N)=k$.  Thus $N'$ has only one blob, and so is itself simple.  However, this does not imply that $N$ is simple.  For instance, $N$ could have many blobs, all but one of which are individually tree-based.  An example might be a network that has one non-tree-based blob of level $\ge 4$, and the other blobs all of level $\le 3$.  Networks (and therefore blobs) of level $\le 3$ are all tree-based (by~\cite{fischer2018nonbinary}), so this network $N$ will have level and tier different, but its minimal support networks will have the same tier and level because they will all be simple.
}

We next turn our attention to another new proximity measure, which can be shown to be bounded by the ones introduced in Section \ref{s:four.equal.measures}.

\subsection{Proximity measure based on adding edges to the network}

Before we can introduce the next measure, recall the method for adding internal edges to a phylogenetic network described in Section~\ref{s:definitions} by choosing two edges, subdividing them with a new vertex each and connecting these two new vertices with an extra edge. 
This procedure does not create any additional leaves and thus also no edges adjacent to leaves, but only internal edges. This immediately leads to the following Lemma. 

\begin{lem}\label{l:add.edges}
If $N\in \PN(X)$
then there exists an 
$N'\in\TBN(X)$
which can be derived from $N$ by adding internal edges. \end{lem}

\begin{proof} Recall that in this manuscript, we only consider proper networks, and that this implies that $N$ either consists of only one vertex (in which case it is trivially tree-based, so there is nothing to show), or has at least two leaves, cf. Section \ref{s:definitions}. So let us assume $N$ is proper and has at least two leaves, but is not tree-based. Then it has a spanning tree with at least one leaf not from $X$. For each such leaf, we claim that one can add an internal edge to the network that allows the creation of a spanning tree with one fewer leaf that is not in $X$.

The method is as follows. Take a spanning tree $T$ for the network and choose a non-$X$ leaf $v$ from $T$ \blue{(see Figure~\ref{f:lev4.edge.add}$(i)$)}.  The degree of $v$ is at least 3 in $N$, so there are at least two edges connected to $v$ in $N$ that are not in $T$. Choose one of them and subdivide it, adding a vertex $v'$.  Now, there is at least one vertex in $T$ of degree at least 3, say $w$, because $T$ has at least 3 leaves (because $N$ has at least 2 and $T$ by assumption has at least one more). Choose an edge incident to $w$ that is on a path that does not lead to $v$ in $T$ (this is always possible because of the at least three edges incident to $w$, only one can be on the unique path from $v$ to $w$ in $T$), and subdivide it, creating a new vertex $w'$.  Now add the edge $e=\{v',w'\}$ to the network \blue{(Figure~\ref{f:lev4.edge.add}$(ii)$)}. Then $T$ can be modified to build a new spanning tree $T'$ by deleting the edge $\{w,w'\}$ and adding the edges $\{v,v'\}$ and $\{v',w'\}$ \blue{(as in Figure~\ref{f:lev4.edge.add}$(iii)$)}. We repeat this procedure until the resulting network $N'$ is tree-based.
\begin{figure}[ht]
\includegraphics{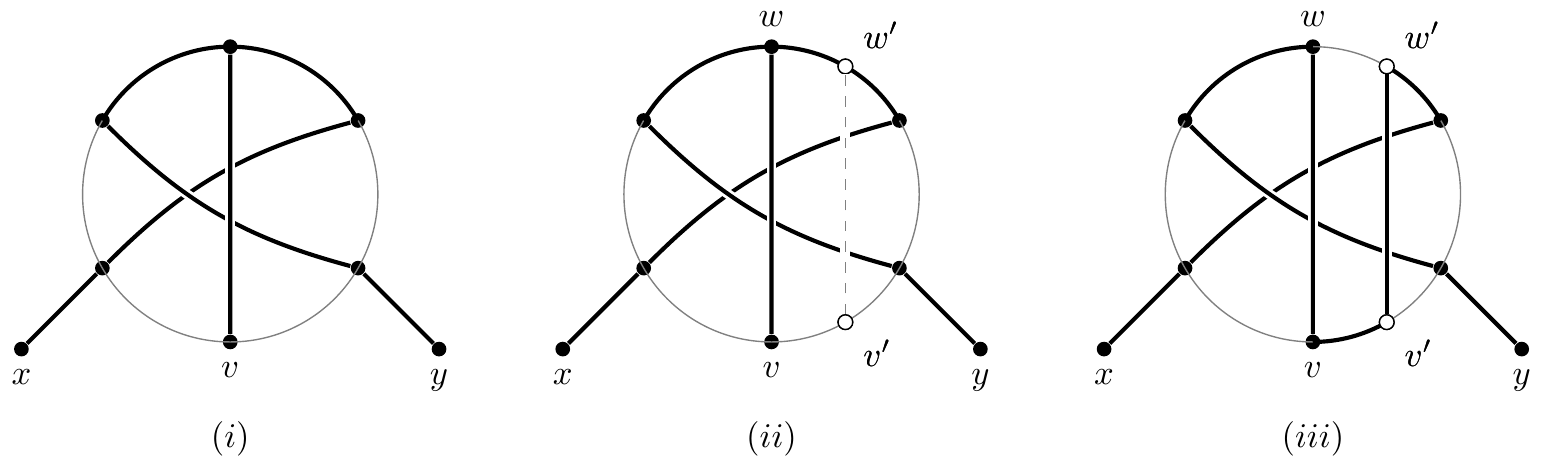}
\caption{\blue{The process used in the proof of Lemma~\ref{l:add.edges} to show that a network can be made tree-based with the addition of internal edges.  This network on $X=\{x,y\}$ is an example of a level 4 network that is not tree-based, and is taken from~\cite[Fig. 9]{fischer2018nonbinary}.  $(i)$ The network with a spanning tree shown in bold, and has a leaf $v\not\in X$.  $(ii)$ Labelling of vertices, as in the proof of Lemma~\ref{l:add.edges}, and an additional internal edge $\{v',w'\}$. $(iii)$ Constructing a new spanning tree by deleting the edge $\{w,w'\}$ and adding the path $\{v,v'\},\{v',w'\}$, which has one fewer leaves. }}
\label{f:lev4.edge.add}
\end{figure}
\end{proof}

We can use Lemma \ref{l:add.edges} to define another proximity measure as follows.
\begin{defn}
For a phylogenetic network $N$, let $e(N)=$ the minimal number of additional edges that need to be added to $N$ to make it tree-based.  
\end{defn}
Note that $e(N)$ is well defined by Lemma~\ref{l:add.edges}, and also note that $e(N)=0$ precisely when $N$ is tree-based. 

Interestingly, $e(N)$ is not always equal to $\ell(N)$ and the other measures in Theorem~\ref{t:equalities}. For an example, see the network in Figure~\ref{f:double.level.5}, which has $\ell(N)=2$ but $e(N)=1$.

We now state our main result about the measure $e(N)$, which provides bounds in the value of $e(N)$ with respect to the measures given by Theorem~\ref{t:equalities}. 

\begin{thm}\label{t:e(N)}
For any $N\in \PN(X)$,
$\frac{1}{2}\ell(N) \le e(N) \le \ell(N)$.
\end{thm}

\begin{proof}
The second inequality follows directly from the proof of Lemma~\ref{l:add.edges}.

For the first inequality, let $N'$ be a tree-based network derived from $N$ by inserting $e(N)$ extra edges, and let $T$ be a support tree of $N'$. 

Observe that $T$ must contain all of the $e(N)$ newly inserted edges, because otherwise if one was not needed, it could be removed from $N'$ without affecting the spanning tree $T$, violating minimality of $e(N)$.
Consequently, if we delete these $e(N)$ edges, $T$ turns into a forest $F$ of $e(N)+1$ components, such that each vertex of $N$ is contained in one of the components of $F$. 

Note that each of the deleted $e(N)$ edges has at most 4 ``connection points'' in $T$, where a connection point is a node adjacent in $N$ to an endpoint of one of the $e(N)$ edges. That is, the connection points are endpoints of edges in $N$ that were subdivided to build $N'$.  If {\it both} vertices of an edge of $N$ that has been subdivided to insert one of the $e(N)$ new edges are connection points in $T$, we call both of them {\it path points}. All connection points that are not path points are called {\it terminals}. 

This implies that we have a forest $F$ covering all nodes of $N$ and in total, we have at most $2 e(N)$ terminals that are leaves in $F$.  This is because when $e$ is removed from $N'$, at most two of the terminals become leaves in $F$ (see the cases in Figure~\ref{f:4endpoints}).  This can happen in Case (i) of the figure, where the top and bottom connection points may possibly become leaves when $e$ is removed (along with the half-edges from the subdivided edges that $e$ was attached to), and in Case (ii) where at most the top connection point could become a leaf when $e$ is removed. 

\begin{figure}[ht]
\includegraphics[width=\textwidth]{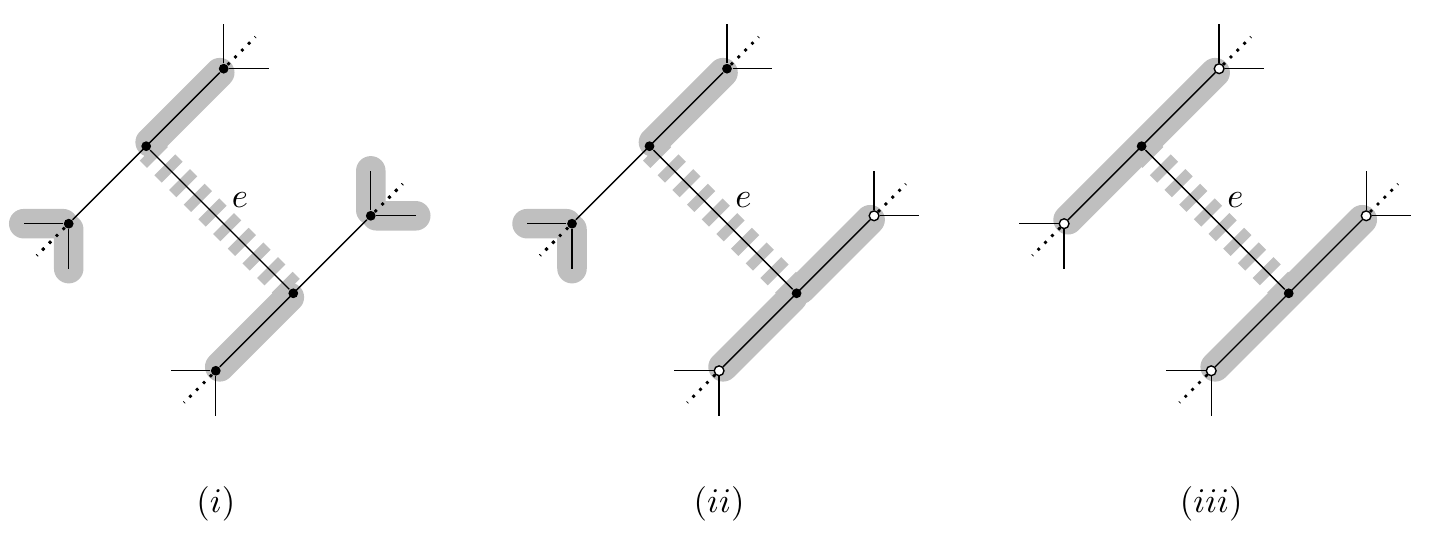} 
\caption{Cases for the spanning tree $T$ in $N'$ using an additional edge $e$, used in the proof of Theorem~\ref{t:e(N)}.  Spanning trees are shown in grey, and path points are indicated by white circles. Case (i) has four terminals and no path points, case (ii) has two terminals and two path points, and case (iii) has no terminals and four path points.}\label{f:4endpoints}
\end{figure}

Note that all trees in $F$ can also be connected via edges in $N$ (as $N$ is connected). So if we now delete the $e(N)$ extra edges from $T$ and use other edges of $N$ to connect the trees of $F$ in $N$, all leaves of the resulting tree $T'$ are either also leaves in $T$ or terminals. As $N'$ is tree-based, all leaves of $T$ are contained in $X$. Therefore, the number of leaves in $T'$ that are not in $X$ is at most the number of terminals, which is at most $2e(N)$. Therefore, $\ell(N)\le 2e(N)$ and thus $\frac{1}{2}\ell(N) \le e(N)$. This completes the proof.
\end{proof}

Note that the bounds given by Theorem \ref{t:e(N)} are tight: examples of reaching the bounds are Figure~\ref{f:double.level.5} for the lower bound, and Figure~\ref{f:level5}(i) for the upper bound. In the latter case, as the network is not tree-based, it can easily be seen that $e(N)=1$ by considering Figure~\ref{f:level5}(ii). Here, this is also equal to $\ell(N)$, which can be seen by connecting one internal endpoint of the paths in Figure~\ref{f:level5}(i) with the adjacent attachment point of a leaf -- this will turn the two depicted paths into a support tree of $N$ with one extra leaf, which is minimal for a non-tree-based network. Therefore, for the network in Figure~\ref{f:level5}(i) we have $e(N)=\ell(N)=1$. 

It should also be noted that Theorem \ref{t:e(N)} connects the intuitively related measures $t(N)$ and $e(N)$ (remembering that $t(N)=\ell(N)$ by Theorem~\ref{t:equalities}): $t(N)$ being the minimum number of leaves one needs to add to make $N$ tree-based, and $e(N)$ the minimum number of edges. 

Moreover, observe that while both $\lambda(N)$ and $e(N)$ are less than $\ell(N)$ (Theorems~\ref{t:e(N)} and~\ref{t:lambda.le.tau}), it is also the case that $\lambda(N)$ and $e(N)$ are not equal in general.  This is essentially because $\lambda(N)$ can be kept small while $e(N)\ge \frac{1}{2}\ell(N)$.  A good example is shown in Figure~\ref{taulambda_fig}, in which we see $\lambda(N)=1$ as long as $m\geq 1$, but $e(N)\ge\frac{m}{2}$ (since $\ell(N)=m$ in this case).

\begin{figure}[ht]
\includegraphics[width=\textwidth]{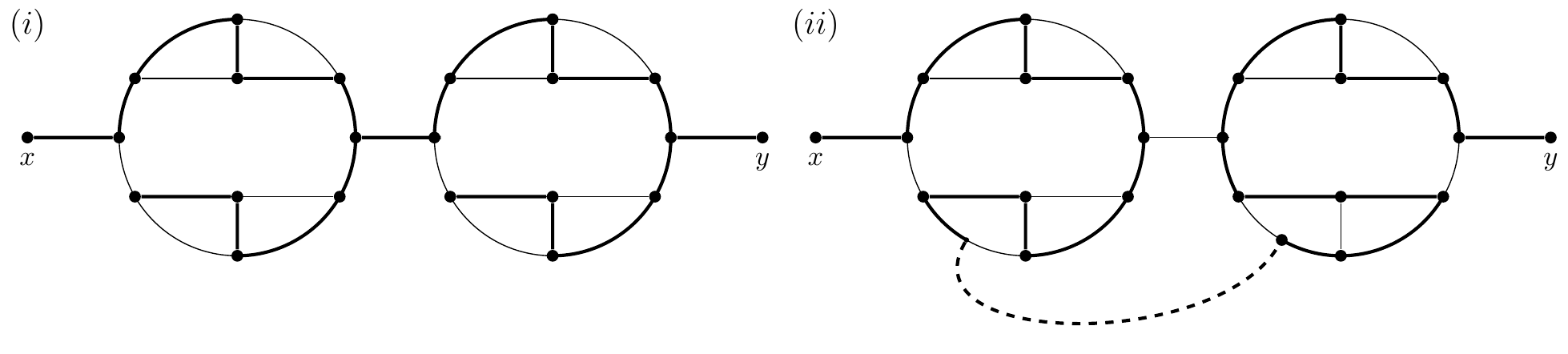}
\caption{$(i)$ A level 5 phylogenetic network $N$ that has $\ell(N)=2$, so that any spanning tree of $N$ has at least two leaves not in $X$; an example spanning tree is shown in bold.  $(ii)$ The same network $N$ with a single additional edge (shown dashed) that makes the network tree-based, showing $e(N)=1$.  A support tree is shown in bold.  }\label{f:double.level.5}
\end{figure}

\subsection{Two more proximity measures linked to the NNI and additional edges} 

We are now in the position to introduce the final two proximity measures and to relate them both to one another as well as the previously introduced measures. One of these new measures, $\delta_{\NNI}$, which is based on NNI moves, is possibly the most interesting one, as already explained in Section \ref{s:intro}. 
Recall the definition of NNI moves from Section~\ref{s:definitions}.

For $N\in \PN(X)$ we define \[\delta_{\NNI}(N)=\min\{d_{\NNI}(N,N')\mid N'\in\TBN(X)\}.\]
Note that because NNI moves do not change the tier (as NNI moves by definition do not change the number of edges of the network), this minimum will be attained for a network $N'$  with $tier(N')=tier(N)$.

A final new proximity measure uses additional edges in spanning trees of the network.

In any spanning tree $T$ of $N$, there is a unique path between any two leaves. If, \blue{for each pair of leaves in $X$,} we consider the union of such paths between leaves, we have a subtree $T_X$ of $T$ that contains the elements of $X$ as leaves, and no other leaves {(note that $T_X$ is the minimum spanning tree of $X$ in $T$)}.  The remaining edges of $T$, namely $E(T)\setminus E(T_X)$, form another set of subtrees of $T$ that we will call ``pendant trees'', each of which has leaves from $N$ plus one leaf that is a vertex from $T_X$ (the attachment point to $T_X$).  

Define a distance to tree-based as follows:
\[m(N):=\min\{|E(T)\setminus E(T_X)|\mid T \text{ a spanning tree of }N\}. 
\]  

This measure has a connection to the NNI distance, and is also greater than or equal to the four proximity measures introduced in Section~\ref{s:four.equal.measures}, as the following result shows.

\begin{thm}\label{t:m(N).is.bigger}
Let $N\in \PN(X)$. Then, we have: 
\begin{enumerate}
	\item $\delta_{\NNI}(N)\le m(N)$, and
	\item $\ell(N)\le m(N)$.
\end{enumerate}
\end{thm}

\begin{proof}\ 

(1) We show that there is an NNI move that reduces $m$ by exactly 1. The idea is then to repeat this move until we reach a tree-based network. 

Suppose that $e$ is an edge in $E(T)\setminus E(T_X)$ which has one vertex $v_i$ on a path between leaves of $N$ and one vertex $w_1$ that is not.  Let $\{v_{i-1},v_i\}$ be an edge in $E(T_X)$ (since $v_i$ is on a path between leaves in $X$ there must be a vertex adjacent to it also on such a path --- it is possible $v_{i-1}$ is itself a leaf in $X$), and let $\{w_1,w_2\}$ be any other edge in $N$ with $w_2\neq v_i$ (such a vertex must exist because $w_1$ is not a leaf in $X$). 

Then the NNI move $v_{i-1},v_i,w_1,w_2\to v_{i-1},w_1,v_i,w_2$ has the effect of inserting the edge $e=\{w_1,v_i\}$ into a path between leaves, replacing the edge $\{v_{i-1},v_i\}$ with the pair of edges $\{v_{i-1},w_1\},\{w_1,v_i\}$ (see Figure~\ref{f:incorp.edge.into.path}).  Since the edge $e$ has shifted from $E(T)\setminus E(T_X)$ to $E(T_X)$, $m$ has decreased by 1 as required. {This shows that we need at most as many NNI moves to reach a tree-based network as there are extra edges in $T$.} 

\begin{figure}[ht]
    \includegraphics{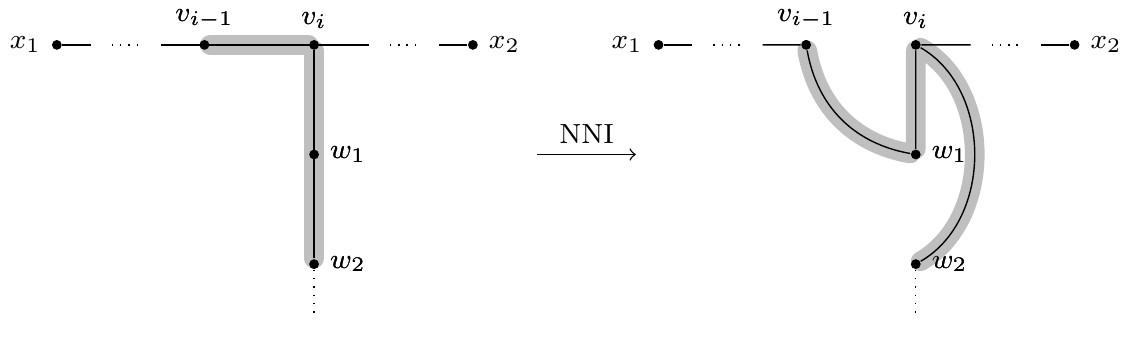}  
    \caption{An NNI move on the path $v_{i-1},v_i,w_1,w_2$ showing how an edge $e=\{v_i,w_1\}$ that is in a spanning tree of $N$ but not on a path between two leaves $x_1,x_2\in X$ can be incorporated into the path.}
    \label{f:incorp.edge.into.path}
\end{figure}

(2)
For a spanning tree $T$ of $N$ that minimizes the size of the set of edges $E(T)\setminus E(T_X)$ (which has size $m(N)$), it will contain as a subset all edges to leaves that are not from $X$.  Thus in particular for that spanning tree the number of such leaves is at most $m(N)$.  Consequently, whatever spanning tree of $N$ that minimizes $\ell(N)$ must have at most $m(N)$ non-$X$ leaves.
\end{proof}

Note that in general $m(N)\neq \delta_{\NNI}(N)$, as the example in Figure~\ref{f:NNI.le.m}(i) shows.  The network shown in Figure~\ref{f:double.level.5} has $m(N)=8$, but $\delta_{\NNI}$ is at most 3. Likewise, in general $m(N)\neq \ell(N)$, with the same example having $\ell(N)=2$.

\begin{figure}[ht]
\includegraphics[width=\textwidth]{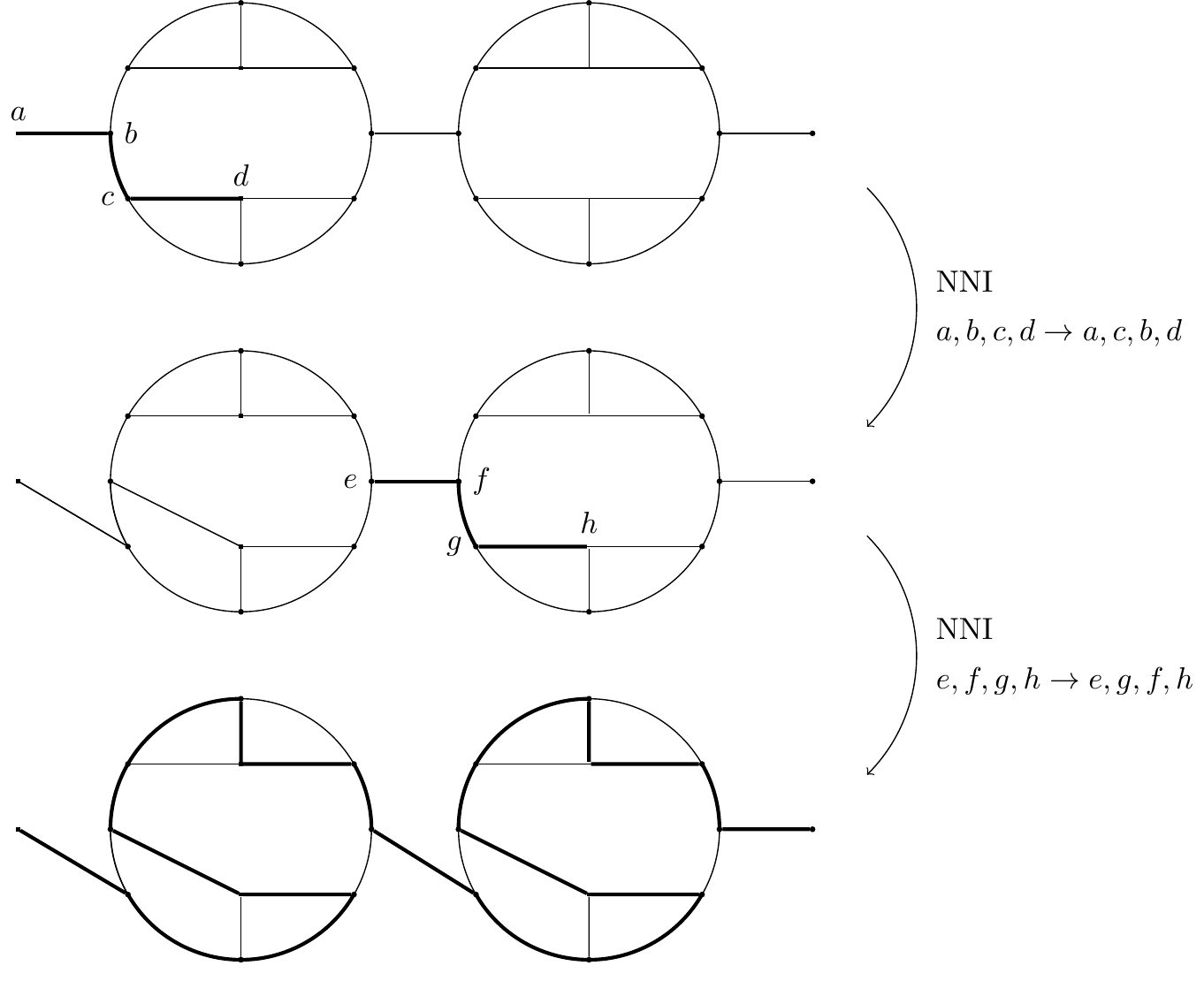}
\caption{This sequence of two NNI moves shows that the distance $\delta_{\NNI}(N)$ from the network $N$ at the top to a tree-based network such as the one at the bottom is at most $2$ 
(in fact, it is precisely 2, because an exhaustive computer search has shown that it is strictly greater than 1).
The paths on which the NNI moves act are shown in bold in the top two figures, and a support tree for the resulting tree-based network at the bottom is also shown in bold.  Only vertices relevant to the NNI moves are labelled. 
Note that $m(N)=8$, as can be seen by the spanning tree shown in Figure~\ref{f:double.level.5}$(i)$, showing that in general $\delta_{\NNI}(N)\neq m(N)$. }
\label{f:NNI.le.m}
\end{figure}

\subsection{\blue{Further questions regarding $\delta_{\NNI}$}}\label{s:NNI.questions}

{The NNI-based proximity measure $\delta_{\NNI}$ gives rise to several interesting additional questions: }

{First, simply regarding the metric, what is the maximum value of $\delta_{\NNI}$ for given $n$ and $k$, where $n$ is the number of leaves and $k$ is the tier? For instance, clearly it is $0$ for $k\le 4$ if $N$ is binary, and for $k\leq 3$ in general (since all such networks are tree-based, cf. \cite{francis2018tree} for the binary case and \cite{fischer2018nonbinary} for the non-binary case) and it is 1 for $k=5$ in the binary case (there are only two non-tree-based networks of level 5, cf. \cite{fischer2018nonbinary}, and both can be made tree-based with a single NNI by changing them to some other network; this works as all others are guaranteed to be tree-based). In all other cases, bounds on $\delta_{\NNI}$ still need to be determined.  

{Second, one might ask an inverse question.  Given 
the large number of possible NNI moves, could it be that 
it is always possible for networks in certain tiers, to make a tree-based network \emph{not} tree-based by a single NNI move? It turns out that in general the answer is ``no'', as Example~\ref{eg:tree-based-rank} illustrates. 

In that light, we could ask just how tree-based is a given tree-based network --- what is the number of moves required to make a tree-based network \emph{not} tree-based? }

The latter two observations suggest a unique measure associated to tree-basedness with NNI moves that may take values of any integer --- positive or negative.  If we define tree-based networks on the boundary of tree-basedness (one NNI move from being a non-tree-based proper phylogenetic network) as having ``tree-based rank'' 0, we can say a tree-based network that is at least $i$ NNI moves from the boundary has tree-based rank $i$, and a non-tree-based proper phylogenetic network that is $i$ NNI moves from being tree-based ($\delta_\NNI(N)=i$) has tree-based rank $-i$.  
That is, we write $||N||_{\TB}$ for the tree-based rank of $N$, defined by
\[||N||_{\TB}=
\begin{cases}
-\delta_\NNI(N) & \text{if }N\in \PN(X)\setminus\TBN(X)\\
|\text{NNI moves to non-tree-based}|-1 & \text{if }N\in\TBN(X).
\end{cases}
\]
Thus, network $N$ in Figure~\ref{f:tier5.rank} has $||N||_{\TB}>0$. It would be interesting to understand what features of a tree-based network give it rank $>0$. 

\begin{example}\label{eg:tree-based-rank}
The network $N$ in Figure~\ref{f:tier5.rank} has $||N||_{\TB}=4$, confirmed with an exhaustive Mathematica search over its 1-, 2-, and 3-neighborhoods. That is, there is a sequence of four NNI moves from $N$ that reach the network shown in Figure~\ref{f:level5}(i), one of the only two non-tree-based binary level 5 proper phylogenetic networks, and all proper phylogenetic networks of less than four NNI moves from $N$ are tree-based.    
Note that one can easily see that $N$ has tree-based rank $>0$ by the following argument. 

Noting that a network obtained by an NNI move on a binary network remains binary, and all proper, binary, phylogenetic networks of level less than 5 are tree-based~\cite{francis2018tree}, for $N$ to be changed to a proper non-tree-based phylogenetic network in a single NNI move, it must become level 5, which requires merging its blobs. 
This forces the single move to be centered on the edge $\{u,v\}$.  The two distinct NNI moves possible centred on this edge produce the networks $N'$ and $N''$ shown: $N'$ is generated by the move $u_1,u,v,v_1$, while $N''$ is generated by the move $u_1,u,v,v_2$ (the other two possible moves are symmetric).  Both of these are tree-based.

While this network has tree-based rank 4, it is nevertheless possible to ``destroy'' its tree-based-ness in just two moves, but only by leaving the space of proper phylogenetic networks. That is, one can perform two NNI moves on $N$ to produce a network that is not tree-based, but it is also not a proper phylogenetic network.  The notion of rank uses a distance \textit{within} the space of proper phylogenetic networks.
\end{example}

\begin{figure}
    \centering
    \includegraphics{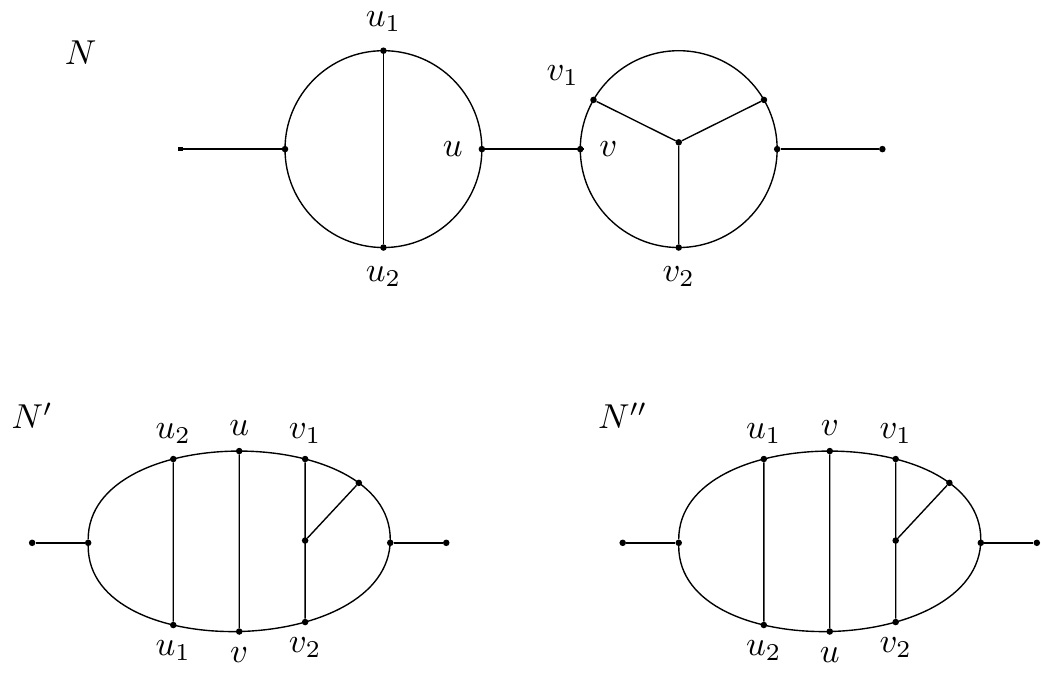}

    \caption{The tier 5 tree-based network $N$ is four NNI moves away from one of the only two non-tree-based binary networks in tier 5 
    (see~\cite{fischer2018classes} for these: one is shown in Figure~\ref{f:level5}). Its two level 5 neighbours, one NNI move away, are $N'$ and $N''$ shown.  See Example~\ref{eg:tree-based-rank}.
    }
    \label{f:tier5.rank}
\end{figure}

We are now in a position to summarize all proximity measures introduced in this manuscript as well as their relationships, before we make some comments about rooted networks. 

\section{A summary of proximity measures introduced in this paper}\label{s:summary}

There have been eight tree-based proximity measures for unrooted phylogenetic networks introduced in this paper, and we summarize them in Table~\ref{tab:summary}.  We have shown the following relationships among these measures:
\begin{itemize}
	\item $\lambda(N)\le \ell(N)=t(N)=p(N)=\tau(N)\le m(N)$ (Theorems~\ref{t:equalities},~\ref{t:lambda.le.tau},~\ref{t:m(N).is.bigger}). 
	\item $\frac{1}{2}\ell(N)\le e(N)\le\ell(N)$ (Theorem~\ref{t:e(N)}).
	\item $\lambda(N)\neq e(N)$ in general (comment after Theorem~\ref{t:e(N)}).
	\item $\delta_\NNI(N)\le m(N)$ and not equal in general (Theorem~\ref{t:m(N).is.bigger}).
\end{itemize}

\begin{table}[ht]
\begin{tabular}{lp{12cm}}
\hline
$\ell(N)$ 	& the minimal number of additional leaves in a spanning tree of $N$\\
$t(N)$		& the minimal number of leaves to add to make $N$ tree-based\\
$p(N)$		& the minimal number of disjoint paths $k$ including a leaf that partition $V(N)$, and for which all but the first {have one} end adjacent to another path; minus $(n-1)$\\
$\tau(N)$	& the minimal tier of a support network of $N$\\
$\lambda(N)$& the minimal level of a support network of $N$\\
$e(N)$		& the minimal number of edges to add to make $N$ tree-based\\
$\delta_\NNI(N)$	& the minimal number of $\NNI$ moves from $N$ to a tree-based network\\
$m(N)$		& the minimal number of edges in a spanning tree of $N$ beyond a spanning $X$-tree\\
\hline
\end{tabular}
\caption{Informal descriptors of proximity measures to tree-based. }\label{tab:summary}
\end{table}

We represent these relationships in Figure~\ref{f:rels}, before we turn our attention to rooted networks. 
\begin{figure}[ht]
\includegraphics[width=0.5\textwidth]{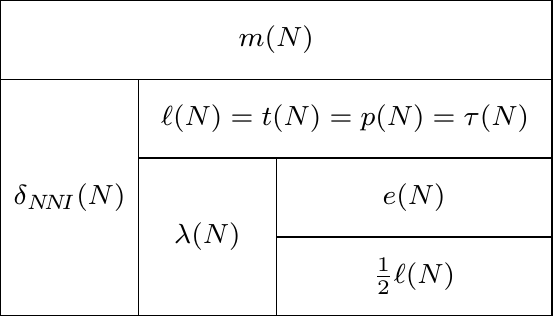}
\caption{Relationships among measures.  The relative position on the vertical axis indicates the inequalities among the measures. For instance, $\lambda(N)\le \ell(N)$.  Measures that are neither above nor below each other are either not linearly comparable, or their relationship is unknown.  For instance, for particular networks $\delta_\NNI(N)$ might be greater than $\ell(N)$ (e.g. Figure~\ref{f:NNI.le.m}), while
the relationship between  $\delta_\NNI(N)$ and $\ell(N)$, or $\lambda(N)$ and $e(N)$, is not known. 
}\label{f:rels}
\end{figure}

\section{Connections with rooted phylogenetic networks}\label{s:rooted}

In this section we discuss some issues surrounding lifting our measures to the rooted situation. We begin by relating tree-basedness in the two contexts.

\subsection{Rooted and unrooted tree-based networks}

First we need to introduce definitions of rooted phylogenetic networks and rooted tree-based phylogenetic networks.

In contrast to unrooted phylogenetic networks, rooted phylogenetic networks have both a special vertex called the root, and an orientation on each edge (hence we call them \emph{arcs}).  That is, a \textit{rooted phylogenetic network} is an acyclic digraph with: a root of in-degree 0 and out-degree at least 1; leaves of in-degree 1 and out-degree 0; and internal vertices of in-degree 1 and out-degree at least 2 (`tree vertices') or out-degree 1 and in-degree at least 2 (`reticulation vertices').  \blue{While it is not uncommon to require the root to have outdegree strictly above 1, we allow degree 1 roots here so that the tree base of such a network, which may have a root of out-degree 1, is also defined as a phylogenetic network.  This is also consistent with the definition of rooted phylogenetic networks in the original introduction of tree-based networks~\cite{francis2015phylogenetic}.}

There are several definitions of rooted tree-based networks in the literature.  Initially they were defined for binary networks on $X$ by saying (roughly) that a network $N^r$ is tree-based if it can be constructed from a tree by additional arcs that avoid cycles, and this was shown to be equivalent to the statement that $N^r$ has a spanning tree whose leaves are those of $X$~\cite{francis2015phylogenetic}: this formulation was used as the definition of tree-based in~\cite{francis2018new}.  The first of these definitions was generalized to the non-binary case in several ways in~\cite{jetten2016nonbinary}, the main one of which is as follows:

\begin{defn}[Definition 4, \cite{jetten2016nonbinary}]\label{d:rooted.treebased}
A rooted nonbinary phylogenetic network $N^r$ is called \emph{tree-based} with base-tree $T$, when it can be obtained from $T$ via the following steps:
\begin{enumerate}
    \item Add some vertices to arcs of $T$ called attachment points, with in-degree and out-degree 1.
    \item Add arcs, called linking arcs, between pairs of attachment points and from tree vertices to attachment points, so that $N^r$ remains acyclic and so that attachment points have in-degree or out-degree 1.
    \item Suppress any attachment points that are not incident to a linking arc.
\end{enumerate}
\end{defn}

For our purposes it makes sense to make explicit a ``spanning tree'' formulation of tree-based networks equivalent to the one for unrooted networks, with the following lemma, \blue{which generalizes the characterization in the binary setting~\cite[Prop 1]{francis2015phylogenetic}.}

\begin{lem}
A rooted phylogenetic network $N^r$ on $X$ is tree-based if and only if it has a (rooted) spanning tree whose leaf-set is $X$.
\end{lem}

\begin{proof}
The reverse direction is immediate: if $N^r$ is tree-based according to Definition~\ref{d:rooted.treebased}, then the tree $T$ on which it is constructed is a spanning tree with leaf-set $X$ as required.

Now suppose that $N^r$ is a rooted phylogenetic network that has a spanning tree $T$ whose leaves are $X$.  We need to show that it is tree-based, which is to say we need to show that it can be constructed from a base-tree via the procedure in Definition~\ref{d:rooted.treebased}.  

The spanning tree $T$ has two kinds of vertices (apart from the root and the leaves): tree vertices (in-degree 1 and out-degree $>1$); and vertices of degree 2.  The arcs that are in $N^r$ but not $T$ are therefore between these types of vertices.  {Now the crucial thing to note is that any tree vertex of $T$ is also a tree vertex in $N$, as its out-degree is larger than 1 in $T$, and thus also in $N$, so its in-degree can only be 1.} Therefore, an additional arc  can never be between a pair of tree-vertices, because there is only one arc into each tree vertex in $N^r$, and so that arc must already be in $T$ (since $T$ is a spanning tree).  For the same reason, such additional arcs cannot be from a degree 2 vertex to a tree vertex. Therefore the arcs that are in $N^r$ but not $T$ are only between degree 2 vertices of $T$, or from tree vertices in $T$ to degree 2 vertices of $T$. Thus $N^r$ is tree-based with base-tree $T$ and attachment points given by the degree 2 vertices in $T$, as required.
\end{proof}

We now characterize the concrete connections between rooted and unrooted tree-based networks.  

For some of these connections, we need to generalize the notion of phylogenetic networks to \emph{degenerate networks}.  

A degenerate network is an acyclic digraph with: a root of in-degree 0 and out-degree at least 1; leaves of in-degree 1 and out-degree 0; and internal vertices of in-degree at least 1 and out-degree at least 2. Note that all rooted phylogenetic networks are contained in the class of degenerate networks, as they fulfill all requirements; but degenerate networks additionally contain networks that have vertices of both in-degree more than 1 and out-degree more than 1 (we call such vertices \emph{degenerate vertices}). We call a degenerate network $D$ \emph{strictly degenerate} if it contains such a vertex. (For a related concept, see the notion of \emph{compressed network}, that removes arcs from reticulate vertices to tree-vertices and replaces each such pair of vertices with a single degenerate vertex~\cite[Section 10.3.4]{steel2016phylogeny}).

{Let $D$ be a strictly degenerate network with degenerate vertex set $\bar{V}$. We call a rooted phylogenetic network $N^r$ a \emph{phylogenetic refinement} of $D$ if $N^r$ can be obtained from $D$ by substituting all vertices $\bar{v} \in \bar{V}$ by two new vertices $\bar{v}_1$ and $\bar{v}_2$ and a directed edge $e=(\bar{v}_1,\bar{v}_2)$ such that all incoming edges of $\bar{v}$ in $D$ are incoming edges of $\bar{v}_1$ in $N^r$ and all outgoing edges of $\bar{v}$ in $D$ are outgoing edges of $\bar{v}_2$ in $N^r$. An example of a phylogenetic refinement of a strictly degenerate network is shown in Figure \ref{Fig_refinement}. The important thing here is to note that \emph{every} degenerate network has a phylogenetic refinement.}

{However, also note that in the binary case, the set of degenerate networks is identical to the set of rooted phylogenetic networks, as in the binary case, for each vertex the sum of its incoming and outgoing edges must be 3; so it is impossible for a vertex to have both indegree and outdegree larger than 1.}

Finally, just as with phylogenetic networks, we call a degenerate network $D$ \emph{tree-based} if it has a spanning tree whose leaf set $X$ coincides with that of $D$. We say that a rooted phylogenetic network (degenerate or not) is \emph{phylogenetically} tree-based if it is tree-based with a support tree whose root has out-degree $\ge 2$.

\begin{figure}[ht]
\includegraphics[width=1\textwidth]{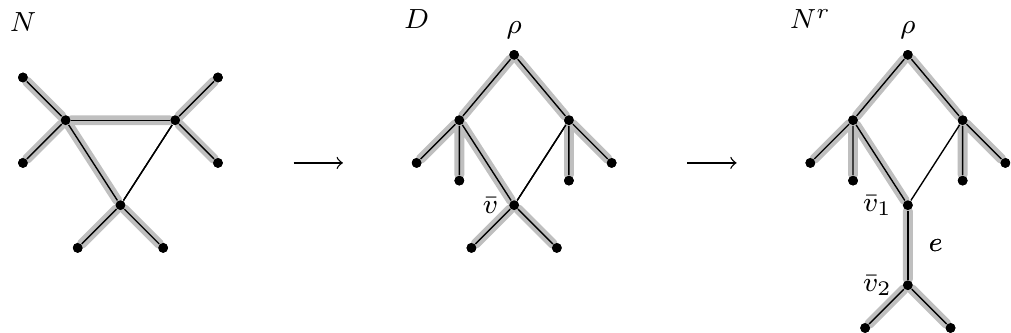}  
\caption{Unrooted tree-based network $N$ and its conversion to a rooted phylogenetically tree-based network $D$, which in this case is strictly degenerate as it contains the degenerate vertex $\bar{v}$. However, $D$ can be converted into its refinement $N^r$ by replacing vertex $\bar{v}$ by a new edge $e=(\bar{v}_1,\bar{v}_2)$. Note that $N^r$ is a rooted phylogenetic network {(as we assume that all edges of $D$ and $N^r$ are directed away from the root)}. 
Support trees for each network are highlighted.}  
\label{Fig_refinement}
\end{figure}

\begin{thm}\label{t:degenerate.rooted-unrooted}
If $N$ is an unrooted phylogenetic network on $X$, then it is tree-based (in the unrooted sense) if and only if it can be rooted on the midpoint of an edge, and with orientations specified on the edges, to give a {degenerate} network $D$ that is phylogenetically tree-based. 

Moreover, if $N^r$ is a phylogenetic refinement of a degenerate (phylogenetically) tree-based network $D$, $N^r$ is also (phylogenetically) tree-based.
\end{thm}
      
\begin{proof}

{We start by proving both directions of the `if and only if' statement.}

($\impliedby$)  
Suppose that $N$ can be made into a  
{degenerate phylogenetically tree-based network } by assigning a root to an edge and specifying orientations on all edges, as in the theorem statement. Then any support tree of $N^r$ whose root has out-degree $\ge 2$ will become a spanning tree for the unrooted network $N$ when the orientations on edges are ignored (suppressing the root vertex if it has degree 2), showing that $N$ is tree-based.

($\implies$)     
If $N$ is tree-based (in the unrooted sense) then we can delete a set $F$ of edges of $N$ to obtain a spanning tree $T$ of $N$ with leaf set $X$.  Now select any edge $e$ of $T$, root $T$ at the midpoint of $e$, and direct all edges of $T$ away from this root.

Now let $t: V \to \{1,2,3,\dots,|V|\}$ be any one-to-one map satisfying the property that if $(w, w')$ is an arc of the directed version of $T$ then $t(w) < t(w')$ (such a map always exists, since by the order extension principle, any poset -- such as the one induced by $T$ --  \blue{has a linear extension}).

Now, for each edge $\{u,v\} \in F$ that was deleted, orient the edge $(u,v)$ (i.e. $u$ directed to $v$)  if $t(u) <t(v)$ and $(v,u)$ if $t(v)<t(u)$. Now all arcs of $N$ are directed in a way that leads to no directed cycles, so we have turned $N$ into a valid {degenerate} network $D$ {(note that $D$ may contain vertices of both in-degree $>1$ and out-degree $>1$, as in Figure \ref{Fig_refinement}, so we cannot guarantee that this $D$ is a rooted phylogenetic network)}, and $N$ is also (phylogenetically) tree-based in the rooted sense (as implied by the rooted version of $T$). This completes the {first part of the }proof.

Next, we show that if $N^r$ is a phylogenetic refinement of a degenerate (phylogenetically)  tree-based network $D$, $N^r$ is also (phylogenetically) tree-based. This can easily be seen, because every support tree $T$ of $D$ is a spanning tree and thus covers all vertices -- in particular also all degenerate vertices. Now, while replacing a degenerate vertex $\bar{v}$ in $D$ by $e=(\bar{v}_1, \bar{v}_2)$ as described above in order to turn $D$ into $N^r$, we also perform the same replacement in $T$ in order to turn $T$ into a support tree $T'$ of $N^r$. In particular, we make sure the edge $e$ is contained in $T'$. This procedure leads to a spanning tree $T'$ of $N^r$ with the same leaf set as $T$, and as both $T$ and $N^r$ have the same leaf set as $D$, this shows that $T'$ is a support tree for $N^r$. This scenario is depicted by Figure \ref{Fig_refinement}. This completes the proof.
\end{proof}     

{This result immediately leads to the following corollary.}

\begin{cor}\label{t:rooted-unrooted} 
If $N$ is a \emph{binary} unrooted phylogenetic network on $X$, then it is tree-based (in the unrooted sense) if and only if it can be rooted on the midpoint of an edge, and with orientations specified on the edges, to give a \emph{binary} rooted phylogenetic network network $N^r$ that is phylogenetically tree-based.\end{cor}

\begin{proof}This is a direct consequence of Theorem \ref{t:degenerate.rooted-unrooted}, exploiting the fact that in the binary case, there exist no strictly degenerate networks, so the conversion of an unrooted network $N$ into a degenerate network $D$ as described in the proof of Theorem \ref{t:degenerate.rooted-unrooted} will immediately lead to a rooted binary phylogenetic network.
\end{proof}

Note that these results explain the contrast in the decision problem between rooted and unrooted networks: for a rooted network, determining whether it is tree-based can be done polynomially~\cite{francis2015phylogenetic}; for an unrooted network, it is NP-complete~\cite{francis2018tree}.  The problem for the unrooted case is that the conversion to a rooted network requires testing orientations on a large number of edges --- an exponential problem.

Moreover, also note that Theorem \ref{t:degenerate.rooted-unrooted} heavily depends on the spanning tree in $N^r$ having a root of outdegree larger than 1; and note that there are rooted phylogenetic networks which are tree-based, but \emph{only} have support trees in which the root has outdegree precisely 1 (for example the right hand network in Figure~\ref{f:rooted.1.leaf}).
This shows that being tree-based is not sufficient for a rooted network to give rise to an unrooted tree-based network --- it indeed needs to be \emph{phylogenetically} tree-based.

{One might ask whether the binary result, Corollary~\ref{t:rooted-unrooted}, extends to non-binary phylogenetic networks without the requirement to include degeneracy.  Using the definition of tree-based in this paper (Hendriksen's ``loosely'' tree-based~\cite{hendriksen2018tree}), the answer is unfortunately no.  This can be seen from considering the tree-based network on six leaves obtained from a triangle with two leaves attached to each corner (network $N$ in Figure~\ref{Fig_refinement})\footnote{We thank Michael Hendriksen for providing this example.}.  There are two alternative definitions of tree-based provided in~\cite{hendriksen2018tree}, one of which does provide a result analogous to Corollary~\ref{t:rooted-unrooted}, as follows.} 

{Hendriksen's definition of ``tree-based'' is that $N$ is tree-based if it has a spanning tree $T$ that has all edges of $N$ between vertices of degree $>3$, and for which all degree 2 vertices in $T$ were degree 3 in $N$.  In this case it is conceivable that a good choice of edge to place a root on might make a network rooted tree-based, but a poor choice might not.  So proving a result like Corollary~\ref{t:rooted-unrooted} may require a way to make wise choices about edges on which to root the network.}  

{The Hendriksen definition of ``strictly'' tree-based, on the other hand, is that $N$ is tree-based if it has a spanning tree that has \emph{all} edges of $N$ incident to vertices of degree $>3$~\cite{hendriksen2018tree}.  In this case, we have the following analogue to Corollary~\ref{t:rooted-unrooted}, that uses Jetten and van Iersel's definition of a strictly tree-based rooted network: $N^r$ is strictly tree-based if it can be obtained from a base tree by adding attachment points to arcs, and connecting additional arcs only between attachment points in such a way as to keep $N^r$ acyclic~\cite[Definition 5]{jetten2016nonbinary}.}  As before, a \emph{phylogenetically strictly tree-based network} is a strictly tree-based network that has a spanning tree whose root has degree greater than $1$.

\begin{thm}\label{t:strictly.rooted.unrooted}
If $N$ is an unrooted phylogenetic network on $X$, then it is strictly tree-based (in the unrooted sense) if and only if it can be rooted on the midpoint of an edge, and with orientations specified on the edges, to give a rooted phylogenetic network $N^r$ that is phylogenetically strictly tree-based. 
\end{thm}

\begin{proof}
{The proof follows very similar lines to that of Theorem~\ref{t:degenerate.rooted-unrooted}.}

If $N$ is an unrooted, strictly tree-based network, then by definition it has a spanning tree $T$ that contains all edges incident to vertices of degree greater than 3.  Thus all edges in $N$ but not in $T$ are between vertices of degree 3 in $N$, and so they connect vertices of degree 2 in $T$. Choosing an edge in $T$, placing a root at its midpoint, and orienting the tree and the edges not in $T$ according to the process described in the proof of Theorem~\ref{t:rooted-unrooted} produces a rooted phylogenetic network that is phylogenetically strictly tree-based as required.  Note that reinserting the edges from $N$ but not $T$  will create degree 3 vertices that cannot be  degenerate, and so the network will be a valid non-degenerate network.

{Conversely, if $N^r$ is a phylogenetically strictly tree-based network then it has a support tree that includes all arcs incident to vertices of degree greater than 3, and with linking arcs between vertices of degree 3.  Ignoring orientation on arcs, and suppressing the root if it is degree 2, we obtain an unrooted phylogenetic network with a spanning tree that contains all edges incident to any vertex of degree greater than 3, and so is a strictly tree-based unrooted phylogenetic network.}
\end{proof}

\subsection{Tree-based proximity, rooted and unrooted}

Proximity measures for rooted phylogenetic networks have already been well-studied~\cite{francis2018new,pons2018tree}, but an interesting question is whether the new measures for unrooted phylogenetic networks introduced in the present paper have analogs in the rooted setting.

At least in one case, the indications are that the lifting is not obvious,  best seen through an example such as the one given in Figure~\ref{f:rooted.1.leaf}.  
This network is not tree-based, and has $\ell(N)=1$ --- there is a spanning tree that has one additional leaf not from $X$ \blue{(see Figure~\ref{f:1-leaf-rooted-eg})}. 

\begin{figure}[ht]
\includegraphics[width=.5\textwidth]{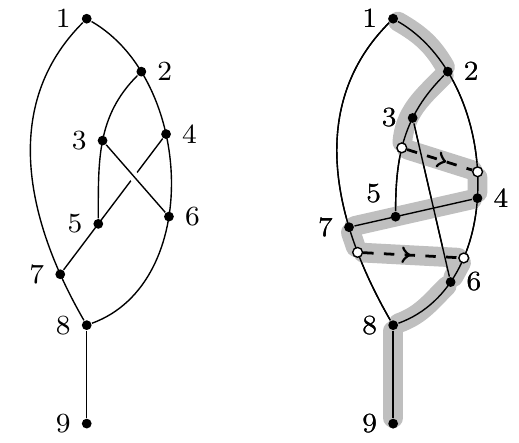}
\caption{On the left is a rooted, non-tree-based phylogenetic network with one leaf, and on the right, the same network with two additional edges (shown as dashed arrows) making it tree-based (with tree base shown in grey).  Note that there are other solutions.  For instance, instead of the lower added edge in the figure, one could add an edge from below 6 to above 5 (subdividing the edge (3,5) twice), also making it tree-based without introducing a cycle. \blue{The same network has tree-based proximity $\ell(N)=1$ (see Figure~\ref{f:1-leaf-rooted-eg}). } }\label{f:rooted.1.leaf}
\end{figure}

However, an exhaustive search using the computer algebra system Mathematica 10 \cite{Mathematica} showed that there is no choice of two edges in $N$ that can be connected by a single edge to make the network tree-based, showing that $e(N)>1=\ell(N)$ (recall that in the unrooted setting we have shown that $\ell(N)\le e(N)$, Theorem~\ref{t:e(N)}). In this search, we checked all 132 combinations of ordered edge pairs such that one was the outgoing edge for the new extra edge and the other one the incoming edge. Note that if the resulting network was tree-based, the support tree would be a Hamiltonian path from the root to the only leaf (as there is just this one leaf). In 128 of the 132 cases, the resulting directed graph had no Hamiltonian path and could therefore not be tree-based. In the remaining four cases, there was a Hamiltonian path from the root to the only leaf, but the graph also contained a cycle -- i.e. it was not a valid rooted phylogenetic network. 

While adding a single edge in this case is not sufficient to make $N$ tree-based, $N$ can be made tree-based with the addition of \emph{two} extra directed edges, meaning that $e(N)=2$. One such choice of two additional edges is shown in Figure~\ref{f:rooted.1.leaf}.  This means that for rooted phylogenetic networks, Theorem~\ref{t:e(N)} does not hold, as here we have $e(N)>\ell(N)$.
 
In general, it seems more difficult to prove that it is always possible to find edges to add to make a network tree-based, largely because placing orientation on the edges means that additional edges can cause cycles.  However, we have failed to find a network for which it has not been possible, and so we make the following conjecture.

\begin{conj}
Let $N$ be a rooted phylogenetic network on leaf set $X$. Then there exists a set of directed edges that can be added to $N$ to turn it into a tree-based rooted phylogenetic network on leaf set $X$. Furthermore, we conjecture that $\ell(N)\le e(N)\le 2\ell(N)$.
\end{conj}

Another interesting aspect about the difference between rooted and unrooted networks is that the proximity measures based on bipartite matchings introduced for rooted networks in \cite{francis2018new} do not carry over to the unrooted case. This is due to the fact that finding a maximum matching and determining its size can be done in polynomial time. However, we know that determining tree-basedness of an unrooted network is NP-complete \cite{francis2018tree}, so the size of a maximum matching cannot be decisive here (unless P=NP).

\section{Discussion}\label{s:discussion}
In the present manuscript, we have generalized --- from the rooted to the unrooted setting --- three proximity measures that measure the distance from any phylogenetic network to a tree-based phylogenetic network. We have shown that -- just like in the rooted case -- these measures turn out to be identical. Moreover, we introduced five new measures, one of which was also identical to the first three, whereas four can be shown to be different in general. Introducing and analyzing such proximity measures for unrooted networks, however, is not only of mathematical interest. It is also relevant for biological studies, where phylogenetic networks are often unrooted, e.g. because the root position is unknown, or because the network represents conflicts in the data rather than the actual evolutionary history of the underlying species. 

Tree-basedness in rooted phylogenetic networks is known to be fundamentally different from the unrooted setting, as it can be decided in polynomial time \cite{francis2015phylogenetic}, whereas it is NP-complete in the unrooted case \cite{francis2018tree}. In this regard, Corollary \ref{t:rooted-unrooted} from the present manuscript will be of wide interest, as it shows that tree-basedness of rooted and unrooted phylogenetic networks is related in an explicit sense: unrooted networks are tree-based if and only if there exists a root position such that the resulting rooted network is tree-based. 

However, note that the difference in the computational complexity of the tree-basedness decision problem between the rooted and the unrooted setting immediately implies that the calculation of \emph{all} proximity measures introduced in the present manuscript for unrooted networks is necessarily NP-complete (otherwise, one could easily determine if these measures are 0, which for all of them is the case if and only if the underlying network is tree-based). So despite the relationships between the measures in the rooted case that have been introduced in \cite{francis2018new} (and which can be calculated in polynomial time) and their generalized unrooted counterparts introduced in the present manuscript, their calculation in the unrooted case is actually hard, except for some classes of networks for which tree-basedness can be guaranteed and for which the proximity measures are therefore necessarily 0~\cite{fischer2018classes}. Therefore, it would be an interesting question for future research to find good approximations to these measures.

\blue{In addition to the questions regarding $\delta_\NNI$ raised in Section~\ref{s:NNI.questions},} other open questions arising from the present manuscript relate to the relationships between the eight measures we introduced in this manuscript: while we have elaborated on some of them and seen, for instance, that some of them are equal and some act as bounds for others, we have not investigated all possible relationships between the different measures. Moreover, in Section \ref{s:rooted} we have discussed some of the issues involved in lifting the new measures to the rooted case.
 
We are confident that these questions will inspire more research, as phylogenetic networks in general and tree-based ones in particular have gained more and more importance over recent years.

\section{Acknowledgements} The authors thank Mike Steel for helpful discussions and personal communication, in particular concerning an idea for the proof of Theorem~\ref{t:degenerate.rooted-unrooted}, and Michael Hendriksen for some comments on a draft.  MF wishes to thank the DAAD for conference travel funding to the Annual New Zealand Phylogenomics Meeting, where partial results leading to this manuscript were achieved. Moreover, Mareike Fischer thanks the joint research project \textit{DIG-IT!} supported by the European Social Fund (ESF), reference: ESF/14-BM-A55-0017/19, and the Ministry of Education, Science and Culture of Mecklenburg-Vorpommern, Germany.


\begin{thebibliography}{10}

\bibitem{fischer2019space}
Mareike Fischer and Andrew Francis.
\newblock The space of tree-based phylogenetic networks.
\newblock {\em arXiv:1910.05679}, 2019.

\bibitem{fischer2018classes}
Mareike Fischer, Michelle Galla, Lina Herbst, Yangjing Long, and Kristina
  Wicke.
\newblock Classes of treebased networks.
\newblock {\em arXiv:1810.06844}, 2018.

\bibitem{fischer2018nonbinary}
Mareike Fischer, Michelle Galla, Lina Herbst, Yangjing Long, and Kristina
  Wicke.
\newblock Non-binary treebased unrooted phylogenetic networks and their
  relations to binary and rooted ones.
\newblock {\em arXiv:1810.06853}, 2018.

\bibitem{francis2018tree}
Andrew Francis, Katharina~T Huber, and Vincent Moulton.
\newblock Tree-based unrooted phylogenetic networks.
\newblock {\em Bulletin of Mathematical Biology}, 80(2):404--416, 2018.

\bibitem{francis2019correction}
Andrew Francis, Katharina~T. Huber, and Vincent Moulton.
\newblock Correction to: Tree-based unrooted phylogenetic networks.
\newblock {\em Bulletin of Mathematical Biology}, 81(3):936--937, Mar 2019.

\bibitem{francis2018new}
Andrew Francis, Charles Semple, and Mike Steel.
\newblock New characterisations of tree-based networks and proximity measures.
\newblock {\em Advances in Applied Mathematics}, 93:93--107, February 2018.

\bibitem{francis2015phylogenetic}
Andrew~R Francis and Mike Steel.
\newblock Which phylogenetic networks are merely trees with additional arcs?
\newblock {\em Systematic Biology}, 64(5):768--777, 2015.

\bibitem{hendriksen2018tree}
Michael Hendriksen.
\newblock Tree-based unrooted nonbinary phylogenetic networks.
\newblock {\em Mathematical Biosciences}, 302:131--138, 2018.

\bibitem{huber2019rooting}
Katharina T Huber, Leo van Iersel, Remie Janssen, Mark Jones, Vincent Moulton, Yukihiro Murakami, and Charles Semple.
\newblock Rooting for phylogenetic networks.
\newblock {\em arXiv:1906.07430v1}, 2019.

\bibitem{huber2016transforming}
Katharina~T Huber, Vincent Moulton, and Taoyang Wu.
\newblock Transforming phylogenetic networks: Moving beyond tree space.
\newblock {\em Journal of Theoretical Biology}, 404:30--39, 2016.

\bibitem{Mathematica}
Wolfram~Research{,} Inc.
\newblock Mathematica, {V}ersion 10.3, 2017.
\newblock Champaign, IL, 2017.

\bibitem{jetten2016nonbinary}
Laura Jetten and Leo van Iersel.
\newblock Nonbinary tree-based phylogenetic networks.
\newblock {\em IEEE/ACM Transactions on Computational Biology and
  Bioinformatics}, 2018.

\bibitem{martin2011early}
William~F Martin.
\newblock Early evolution without a tree of life.
\newblock {\em Biology Direct}, 6(1):36, 2011.

\bibitem{pons2018tree}
Joan~Carles Pons, Charles Semple, and Mike Steel.
\newblock Tree-based networks: characterisations, metrics, and support trees.
\newblock {\em Journal of Mathematical Biology}, pages 1--20, 2018.

\bibitem{steel2016phylogeny}
Mike Steel.
\newblock {\em Phylogeny: discrete and random processes in evolution.}
\newblock SIAM, 2016.

\bibitem{zhang2016tree}
Louxin Zhang.
\newblock On tree-based phylogenetic networks.
\newblock {\em Journal of Computational Biology}, 23(7):553--565, 2016.

\end{thebibliography}
\end{document}